\documentclass[11pt, reqno]{amsart}
\usepackage{amsmath,amssymb,amsthm,amsfonts,amscd, url}
\usepackage{hyperref}
\usepackage[numeric]{amsrefs}
\usepackage{subcaption}
\input xy
\xyoption{all}
\usepackage[labelfont=bf]{caption}
\usepackage{float}

\usepackage[dvipsnames]{xcolor}
\usepackage{graphics}
\usepackage{graphicx}
\usepackage{amsthm}

\newcommand\numberthis{\addtocounter{equation}{1}\tag{\theequation}}
\DeclareMathOperator{\spn}{span}

\newcommand{\tpoint}[1]{\vspace{3mm}\par \noindent \refstepcounter{subsection}{\bf \thesubsection.}
  {\textit{#1} }}

\newtheorem{theorem}{Theorem}

\newtheorem{nlem}{Lemma}
\newtheorem{nprop}{Proposition}

\theoremstyle{remark}
\newtheorem{define}{Definition}
\numberwithin{define}{subsection}
\newtheorem*{nrem}{Remark}
\newtheorem*{notation}{Notation}

\numberwithin{equation}{section}
\numberwithin{nlem}{section}
\numberwithin{theorem}{section}
\numberwithin{nprop}{section}
\numberwithin{ncor}{section}

\DeclareMathOperator{\Ima}{Im}

\newcommand{\be}[1]{\begin{eqnarray} \label{#1}}
\newcommand{\ee}{\end{eqnarray}}

\title{The analysis of topological structure in data using persistent homology; applications to lexical word association networks}

\author{Matthew Pietrosanu, University of Alberta}
\date{\today}

\begin{document}

\begin{abstract}
Persistent homology is a technique recently developed in algebraic and computational topology well-suited to analysing structure in complex, high-dimensional data. In this paper, we exposit the theory of persistent homology from first principles and detail a novel application of this method to the field of computational linguistics. Using this method, we search for clusters and other topological features among closely-associated words of the English language. Furthermore, we compare the clustering abilities of persistent homology and the commonly-used Markov clustering algorithm and discuss improvements to basic persistent homology techniques to increase its clustering efficacy.
\end{abstract}
\maketitle
\section{Introduction\\}

\tpoint{Background and recent history of topology\\}

The history of topology begins with development of algebraic topology in a series of papers published in 1894 and 1895 by Henri Poincar\'{e} (\cite{tophistory}, Preface). This field of mathematics examines properties of geometric objects invariant under continuous, invertible transformations such as stretching or bending, called \textit{homeomorphisms}. These invariant properties include the notions of connectivity and \textit{genus}---informally, the number of ``holes'' in an object---and are irrespective of scale, shape, and any underlying coordinate system. In contrast, classical Euclidean geometry only considers the so-called \textit{rigid transformations} of translation, rotation, and reflection. As such, topology is much less strict in its classification of geometric bodies than is geometry. A classical and well-known example of the generality of topological classification is the equivalence of a coffee mug and a doughnut, as shown in \textbf{Figure \ref{coffee}}. These two objects are certainly not equivalent under the rigid transformations of Euclidean geometry.
\begin{figure}[!h]
\centering
\includegraphics[scale=0.40]{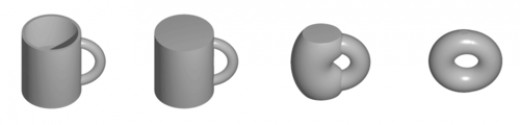}
\caption{\footnotesize{The transformation of a coffee mug into a doughnut under a continuous, invertible transformation. Intuitively, the transformation is continuous because it does not tear the object, and it is invertible because the transformation can be reversed.}}\label{coffee}
\end{figure}

Though rooted and grown in the realm of pure mathematics for most of its history, topology has recently piqued interest across numerous disciplines, including the biological \cite{bio}, medical \cite{med}, and computational \cite{image} sciences. Advances in these fields have necessitated the development of new methods for data analysis and visualisation, particularly where collected data is of high dimension---that is, measured by a large number of descriptors or independent variables---and difficult to interpret. Carlsson, a major contributor to the development of persistent homology, argues that metric- and coordinate-based analytic methods, specifically outside the field of physics, are often unjustified, subjective, and unnatural, especially where analysis is exploratory and intended to yield preliminary or qualitative results (\cite{topdata}, Section 1). As such, topology, with its classification flexibility and disregard of coordinates and metrics, presents itself as a useful data analytic tool. The application of topology to other areas has given rise to the field of \textit{computational topology}.\\

\tpoint{Persistent homology\\}\label{sec.ph}

Foundational to the methods of computational topology is \textit{persistent homology}. This technique makes use of data sampled from some unknown object, space, or phenomenon to recreate the original object's structure and approximate its topology. This recovery of topological properties has direct application to computer image processing \cite{image}, 3D-modeling, and network analysis \cite{sensor} in physics and computing science. In fields such as biology where data may not necessarily be sampled from a physical object, persistent homology can yield insight into the processes and phenomena underlying and creating the observed data.\\

This topological reconstruction is carried out by building structures, called \textit{simplicial complexes}, using collected data. Although there are numerous ways to construct a simplicial complex, these complexes generally serve to join data points that are deemed sufficiently ``close'' under some pre-specified notion of similarity that is in general not a metric. By observing the \textit{homology}---informally, the connected components, loops, and ``holes''---of a simplicial complex, we approximate the homology of the original, unknown object. Simplicial complexes can be created on any scale, effectively permitting the examination of the original object at any local or global scale or resolution.\\

A fundamental idea underlying persistent homology is that true, global features of the original, unknown object will be present in the simplicial complexees across a wide range of scales. Local features and random noise, on the other hand, will appear only over a limited range of scales. The global significance of a topological feature observed in a simplicial complex, then, is determined by the range of scales over which the feature exists, and is referred to as the feature's \textit{persistence} or \textit{lifetime}. The persistence of all observed features can be represented mathematically by a collection of intervals, and visually by a series of lines, referred to as a \textit{barcode}. Using this barcode, we can make inferences regarding the number of topological features of each dimension present in the original object, referred to as the object's \textit{Betti numbers} or, more precisely put, the dimension of the object's homology groups.\\

The nature of a topological feature is determined by its dimension: zero corresponds to connected components; one to loops, such as those of a circle or torus; two to voids, such as the space enclosed by a sphere; and so on. The number of zeroth-dimensional features are of particular significance in statistics and machine learning, as \textit{clustering}---the grouping of data based on some notion of similarity---is often required in these fields for pattern recognition and general data analysis. Clustering corresponds to the problem of finding connected components in persistent homology. See \textbf{Figure \ref{cluster}} for a visual example of clustering.\\

\tpoint{Linguistic applications\\} \label{linguisticapplications}

As an example to be carried throughout this paper, consider the words of the English language. Each word carries with it a certain related concept, idea, or notion: certain pairs or sets of words may overlap in the ideas associated with them. As a result, a given word may be more closely mentally-associated with one word than another. An example of the large-scale structure that this kind of association can form among words, called a \textit{word association network}, is given in \textbf{Figure \ref{sad}}. This notion of similarity between words, hereafter referred to as \textit{association strength}, allows a word association network to be divided into clusters by persistent homology or other clustering algorithms.\\

A knowledge of how the words of a language cluster together has numerous implications for research and everyday life (\cite{ling1}, Section 6). Such an understanding can suggest new experiments in psychology and psycholinguistics to investigate, for example, how the association of various concepts changes, grows, or degrades during childhood development or with increasing age. In artificial intelligence, the application of word clusters could aid in context recognition for both written and spoken language. Furthermore, electronic dictionaries could be made friendlier to language learners by listing closely-associated words and phrase patterns.\\

\begin{figure}[t]
    \begin{subfigure}[H]{0.5\textwidth}
        \centering
        \includegraphics[scale=0.5]{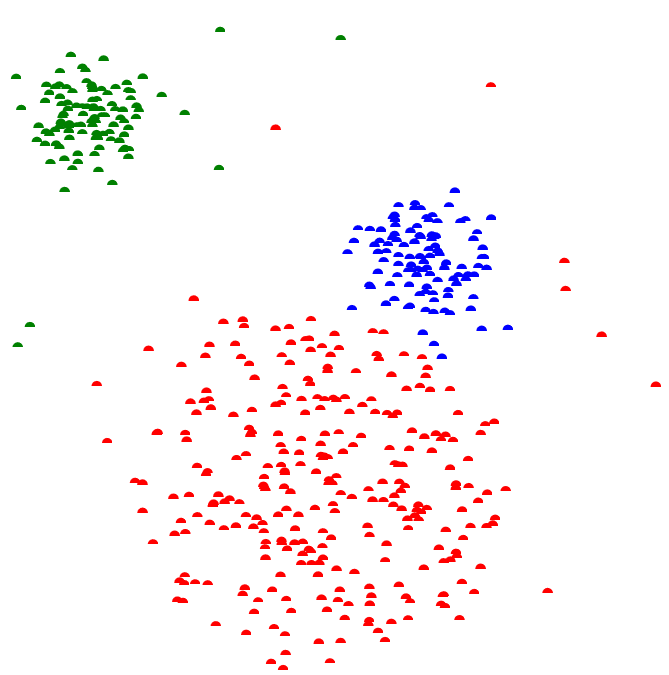}
        \caption{\footnotesize{An example of data clustering. Each symbol represents a data point; points assigned to the same cluster are represented by the same colour \cite{cluster.fig}.}}\label{cluster}
    \end{subfigure}
    ~ 
    \begin{subfigure}[H]{0.5\textwidth}
        \centering
        \includegraphics[scale=0.19]{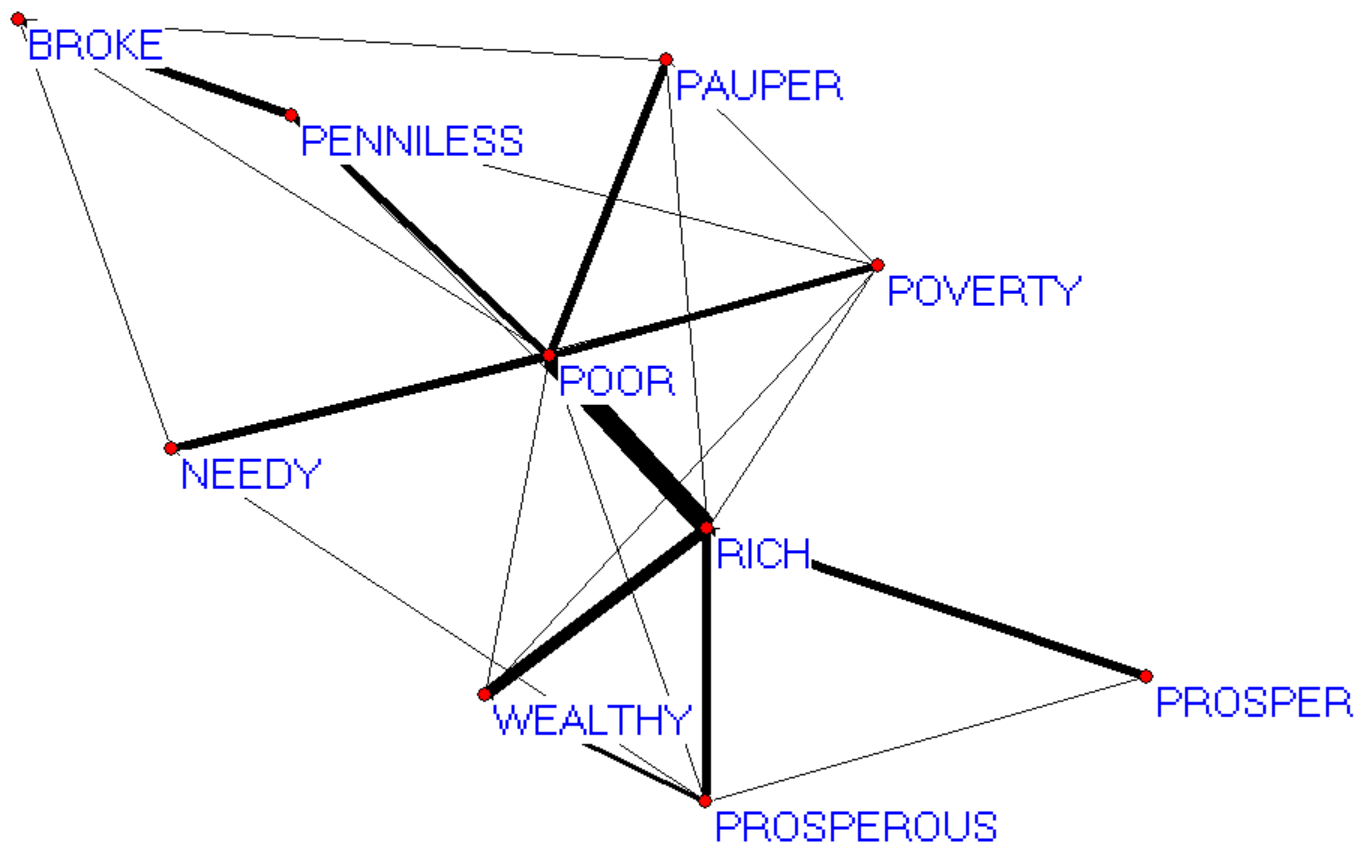}
        \caption{\footnotesize{A small group of words and the strength of the associations between them. The thickness of a line between words is representative of association strength.}}\label{sad}
    \end{subfigure}
    \caption{}
\end{figure}

\tpoint{Objectives and methods\\}

In this paper, we exposit the theory and method of persistent homology from first principles to the topics of simplicial complex construction, homology groups, Betti numbers, and persistence. We then detail the application of persistent homology to finite sets of data. For ease of visualisation, only data representable in some Euclidean space is considered, and we use the Vietoris-Rips complex construction for its computational efficiency. We note, however, that the methods presented generalise readily to data in any space as well as to other simplicial complex constructions. The theoretical portion of this paper loosely follows select sections from the text by Edelsbrunner (\cite{comptopol}).\\

Additionally, we detail a novel application of persistent homology to linguistics for the purpose of finding clusters of closely-associated English-language words. We compare the clustering abilities of persistent homology against Markov Clustering---an algorithm that has previously been applied to large-scale word association networks (\cite{ling1}, Section 4)---and use the standard graph theoretic modularity index \cite{modularity} to assess the quality of the clusters generated by each method.\\

In addition to the clusters found in the Edinburgh Associative Thesaurus (EAT) \cite{eat} by persistent homology, we present our results for higher-dimensional features. In particular, we include visual examples and offer interpretations of these features in a linguistic context.\\

The data contained in the publicly-available EAT was used in our investigation to compute association strengths between words. The program R, together with the TDA package for topological data analysis \cite{rtda}, as well as van Dongen's Markov Clustering algorithm and code \cite{mcl} were the major tools employed in our analysis. Furthermore, the Pajek Program for Large Network Analysis \cite{pajek} was used to create the visualisations presented in this paper unless otherwise noted.\\

\tpoint{Acknowledgements\\}

We acknowledge the early work of Herbert Edelsbrunner, Afra Zomorodian, Gunnar Carlsson, Robert Ghrist, and Peter Bubenik in developing the field of computational topology. I thank my supervising professor, Giseon Heo, for her guidance throughout this project, and Jisu Kim, among the authors of the R-TDA package, for his advice and technical assistance.\\

\tpoint{Structure of this paper\\}

We present the theory of persistent homology in Sections \ref{simplexsection} through \ref{simplicialhomologysection}. Section \ref{simplexsection} introduces the basic construction of \textit{simplices}, \textit{simplicial complexes}, and \textit{filtrations} on point-cloud data; in particular, we use the \textit{Vietoris-Rips} complex construction. Section \ref{boundscyclessection} develops the \textit{chain}, an algebraic structure on simplicial complexes that underlies persistent homology. We make specific note of a mapping between chains of different dimensions as well as various properties of this map, called the \textit{boundary operator}. Section \ref{simplicialhomologysection} presents \textit{simplicial homology}, the theory of homology groups in the setting of simplicial complexes. Also included is an explicit example demonstrating the calculation of a homology group for a given simplicial complex. In a final theoretical subsection, we briefly define of the \textit{Betti number}, an important numerical summary of a homology group, and present the notion of \textit{barcodes} and \textit{persistence}.\\

Our application of persistent homology to the EAT is detailed in Sections \ref{sectioncomponents} through \ref{interpretationsection}. Section \ref{sectioncomponents} introduces the EAT and defines other topics prerequisite for the proposed analysis, such as the \textit{Markov Clustering algorithm} and the \textit{modularity index} for assessing clustering quality. Clustering results are presented and discussed in Section \ref{resultssection}, with suggestions given for ways to improve the clustering efficacy of persistent homology. Lastly, Section \ref{interpretationsection} briefly examines specific clusters and other topological features found using persistent homology and offers an interpretation of these features in a linguistic context.\\

Following the main body is this paper are two appendices. Appendix Section \ref{appendixproofs} provides algebraic proofs and definitions deemed too technical for the main discussion. Appendix Section \ref{appendiximages} includes additional images of the topological features found in the EAT data using persistent homology.

\section{Simplicial Complexes\label{simplexsection}\\}

In this section, we develop the tools necessary to convert a finite set of points into objects containing information about the topology of the space from which the points were sampled. We build these objects, called simplicial complexes, up from their constituent parts using the Vietoris-Rips complex construction.\\

Our general goal and motivation, as in most statistical investigations, is to elucidate patterns and structure present in a given set of data. This data, when representable in some finite-dimensional Euclidean space, is called a \textit{point-cloud dataset}. In this section, we assume that all points are elements of a fixed, finite-dimensional Euclidean space.\\

\tpoint{Basic simplicial structure\\ \label{basicsimplex}}

We say that a set of points $\{x_0,x_1,\dots,x_k\}$ is \textit{affinely independent} if the set $\{x_i-x_0 \mid 1 \leq i \leq k\}$ is linearly independent. Essentially, affine independence redefines the usual linear independence of vectors by using a fixed, arbitrarily-chosen point of the set as the origin. In the above definition, we use $x_0$ to represent this new origin, although the choice of point in the set is independent of the set's affine independence. It is from sets of affinely independent points that we build basic simplicial structures.\\

\begin{define} \label{simplex}
For some non-negative integer $k$, define the \textit{$k$-simplex} corresponding to a set of $k+1$ affinely independent points $\{x_0,x_1,\dots,x_k\}$ to be the set of all linear combinations of the form $\sum_{i=0}^{k}\lambda_ix_i$, with $\lambda_i$ non-negative for all $i$ and $\sum_{i=0}^{k}\lambda_i=1$. Such a simplex will be denoted $\sigma_{\{x_0,x_1,\dots,x_k\}}$ or $\sigma$ where context is clear.\\
\end{define}

In this setting of Euclidean space, the simplices take on familiar forms: for $k=0,1,2,3$, a $k$-simplex is a point, line, closed triangular region, and solid tetrahedron, respectively, and as shown in \textbf{Figure \ref{eg-simplex}}. More generally, the $k$-simplex corresponding to a set of points is the smallest convex set containing the given points.\\

Note that we require affine independence in the above definition to preclude degenerate simplices from forming, such as the degenerate 2-simplex with all three of its vertices on a single line. Ultimately, affine independence prevents any three points from lying on the same line, any four points from lying in the same plane, and so on.\\

\begin{figure}[!h]
  \centering
    \includegraphics[scale=0.75]{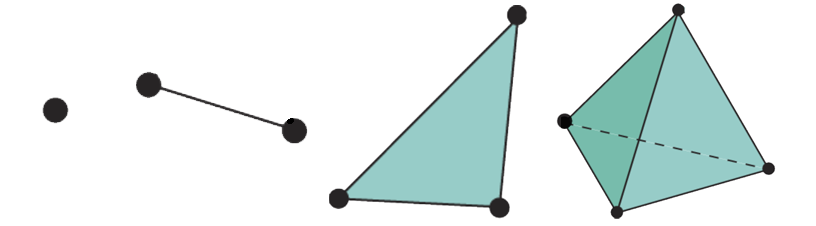}
  \caption{\footnotesize{$k$-simplices for $k=0,1,2,3$, respectively.}}\label{eg-simplex}
\end{figure}

We can also examine the sub-simplices that make up a given simplex. Following the notation of Definition \ref{simplex}, the notion of a \textit{face} can be introduced.\\

\begin{define} \label{face}
Let $\sigma_A$ be the simplex corresponding to the set $A$ of affinely independent points. We say that $\sigma_B$ is a face of $\sigma_A$ if $B$ is a subset of $A$.\\
\end{define}

For example, the faces of a 3-simplex---that is, a solid tetrahedron---consist precisely of the tetrahedron's four triangular sides, six edges, and four vertices.\\

Now equipped with the notion of simplices and faces, we can create increasingly-complex structures by gluing multiple simplices together.\\

\begin{define}
A \textit{simplicial complex} $X$ is a finite collection of simplices, satisfying the following conditions:
\begin{enumerate}
\item For every simplex $\sigma$ in $X$ and every face $\tau$ of $\sigma$, $\tau$ is also in $X$.
\item For any distinct simplices $\sigma_1$ and $\sigma_2$ in $X$, either $\sigma_1\cap\sigma_2$ is empty or is a face of both $\sigma_1$ and $\sigma_2$.
\end{enumerate}
\end{define}

Informally put, a simplicial complex $X$ contains the faces of all its simplices, and intersects simplices only along entire faces.

\tpoint{Constructions on point-clouds\\} \label{pcconstructions}

By building a simplicial complex from a point-cloud, we can begin to examine the topological properties of the space the data was sampled from. These topological properties include the number of connected components and ``holes'' of any dimension, such as the loops of an $n$-fold torus or the void enclosed by a sphere. While there are numerous ways to construct simplicial complexes on a given point-cloud, we present one method used widely for its computational efficiency.\\

In this subsection, we continue to use the notation of Definition \ref{simplex}.\\

\begin{define}\label{VR-Complex}
Let $P$ be a point-cloud in a Euclidean space equipped with some metric $d(\cdot,\cdot)$. Fix $\varepsilon\geq0$ a non-negative, real number. We construct the \textit{Vietoris-Rips simplicial complex} of radius $\varepsilon$ on $P$, denoted $V_\varepsilon(P)$, according to the following rules:
\begin{enumerate}
\item The 0-simplices of $V_\varepsilon(P)$ are taken to be the points of $P$.
\item Given $x$ and $y$ in $P$, the 1-simplex $\sigma_{\{x,y\}}$ is in $V_\varepsilon(P)$ if and only if $d(x,y)$ is at most $\varepsilon$.
\item If $A$ is a subset of $P$, the simplex $\sigma_A$ is in $V_\varepsilon(P)$ if and only if all faces of $\sigma_A$ are also in $V_\varepsilon(P)$.\\
\end{enumerate}
\end{define}

While $P$ is specified by the given point-cloud, the parameter $\varepsilon$ is free to be chosen arbitrarily. Intuitively, $\varepsilon$ acts as a tuning parameter to adjust the ``coarseness'' of the resulting Vietoris-Rips complex. Here, a pair of points are joined with an edge---that is, a 1-simplex---if and only if those points are within distance $\varepsilon$ of each other. By condition (3), a higher-order simplex is added to the complex only where all faces of the simplex are already present in the complex.\\

Given a point-cloud $P$, we can then construct a family of Vietoris-Rips complexes indexed by a single parameter $\varepsilon$, called a \textit{filtration} of complexes. As $\varepsilon$ increases, new topological features, such as connected components and loops, are created, and existing topological features become connected with other features.\\

Loosely speaking, the range of $\varepsilon$ for which a given topological feature exists is indicative of the feature's significance in the sample space. Features that persist for a wide range of $\varepsilon$ are likely to represent true, global features of the sample space, while those that disappear quickly are likely to be either local features or simply noise in the point-cloud data. We revisit and formally define this notion of feature significance in Section \ref{persistencebarcodes}.\\


\section{Boundaries and Cycles\\} \label{boundscyclessection}

In the following section, we continue to develop the theory of persistent homology by building algebraic structures, called \textit{chains}, on simplicial complexes. In particular, we focus on two kinds of chains central to homology, called \textit{boundaries} and \textit{cycles}, and examine a particular map relating the two, called the \textit{boundary operator}.\\

We make use of a number of well-known group theoretic results presumably present in any introductory-level text. Proofs of elementary claims are included Appendix Section \ref{appendixproofs}, while external references will be made to other works for advanced results outside the scope of this paper. We explicitly reference Goodman's abstract algebra text \cite{goodman} and loosely follow Munkres' algebraic topology text \cite{munkres}.\\

Our general goal and motivation in the following section is, informally, to develop the notion of a loop or cycle in a simplicial complex. The first step in doing so is to formalise the direction, or \textit{orientation}, of such loops. As in previous sections, we consider a single simplex before generalising to simplicial complexes. The following subsection recalls some prerequisite concepts from elementary group theory.\\

\tpoint{Permutations\\ \label{permutationsubsection}}

Recall that a \textit{permutation} of a finite set is a bijection from that set to itself. For example, one permutation of the set $\{1,2,3\}$ is the bijection $\pi_1$ that maps 1 to 3, 3 to 1, and 2 to itself. In other words, $\pi_1$ maps the sequence $(1,2,3)$ to $(3,2,1)$, as shown in \textbf{Figure \ref{pi1}} below. Of particular note are the \textit{transpositions}, that is, permutations that interchange exactly two elements, as in the example just given.

\begin{figure}[h]
$\pi_1: (1,2,3)\xrightarrow{1\leftrightarrow3}(3,2,1)$
\caption{\footnotesize{The permutation $\pi_1$ of the set $\{1,2,3\}$ as introduced above. Note that $\pi_1$ is a transposition because it only switches two elements of the sequence $(1,2,3)$, namely, 1 and 3.}} \label{pi1}
\end{figure}

Define an \textit{ordering} of a finite set $S$ to be an ordered sequence $(x_0,x_1,\dots,x_k)$ of the elements of $S$ in which every element of $S$ appears exactly once. The permutation corresponding to such an ordering is the permutation on $S$ that maps $x_i$ to $x_{i+1}$ for $i=0,1,\dots,k-1$, and maps $x_k$ to $x_0$. In other words,
\begin{equation*}
x_0 \mapsto x_1\ \mapsto \dots x_i \mapsto x_{i+1} \mapsto \dots \mapsto x_n \mapsto x_0,
\end{equation*}
where the arrows represent the mappings of this permutation. Intuitively, the permutation corresponding to an ordering simply ``cycles through'' its elements. For our purposes later on, an ordering will, intuitively-speaking, specify a ``path'' visiting all the vertices of a given simplex.\\

Now let us consider the following example and corresponding \textbf{Figure \ref{123eg}}. Let $\pi_2$ be the permutation of the set $\{1,2,3\}$ that maps 1 to 3, 2 to 1, and 3 to 2. Observe that $\pi_2$ can be viewed as a series of transpositions, first switching 3 with 2, and then switching 1 with 3. Thus we see that $\pi_2$ can be written as the composition of an even number of transpositions.\\

\begin{figure}[h]
$\pi_2: (1,2,3) \xrightarrow{2\leftrightarrow3} (1,3,2) \xrightarrow{1\leftrightarrow3}(3,1,2)$
\caption{\footnotesize{Overall, the above series of transpositions is equivalent to the permutation $\pi_2$ by mapping 1 to 3, 2 to 1, and 3 to 2.\\}}\label{123eg}
\end{figure}

By a well-known result of group theory, this result holds in general: any given permutation of a finite set with at least two elements can be represented as a composition of transpositions. Furthermore, although such a representation is not unique, the number of transpositions used to compose a given permutation will either be invariably even or odd (\cite{goodman}, Section 2.4).\\

In the example of \textbf{Figure \ref{123eg}}, observe that both $(1,2,3)$ and $(3,1,2)$ are orderings of the set $\{1,2,3\}$. As noted above, the permutation $\pi_2$ mapping $(1,2,3)$ to $(3,1,2)$ can be represented by the composition of two transpositions---an even number. Then by the above result, any composition of transpositions mapping $(1,2,3)$ to $(3,1,2)$ must also use an even number of transpositions. We then say that the these two orderings \textit{differ by an even number of transpositions}.\\

In general, we say that two orderings of the same set differ by an even number of transpositions if the permutation mapping one ordering to the other can be written as the composition of an even number of transpositions. Otherwise, we say that the two orderings \textit{differ by an odd number of transpositions}.\\

\tpoint{Oriented simplices\label{orientedsimplicessection}\\}

We next consider orientations of a simplex relative to an ordering of its vertices. Orientation constitutes a subtle yet necessary part of the algebraic structure we will soon impose on simplicial complexes. Here we will assume $\sigma_S$ to be a simplex on $S$ after the notation of Definition \ref{simplex}.\\

\begin{define}\label{orientedsimplex}
For $k$ strictly positive, an \textit{oriented $k$-simplex} is a $k$-simplex $\sigma_S$ together with an ordering of $S$. We say that two orderings of $S$ are of the \textit{same orientation} if and only if the two orderings differ by an even number of transpositions. Furthermore, two oriented $k$-simplices $\sigma_S$ and $\tau_S$ are said to be of the same orientation if their orderings differ by an even number of transpositions.
\end{define}
\begin{nrem}
Recall the result of the previous subsection stating that any permutation on a set of at least two elements can be decomposed into transpositions. Observe that this statement does not hold for singleton sets. Indeed, the only permutation on such a set is the identity map that can be written as the repeated composition of itself any even or odd number of times.\\

For this reason, we define an \textit{oriented 0-simplex} to be a 0-simplex with no orientation.\\
\end{nrem}

We now extend the notion of an oriented simplex to the set of $k$-simplices of a simplicial complex.

\begin{figure}[!h]
\centering
\includegraphics[scale=0.75]{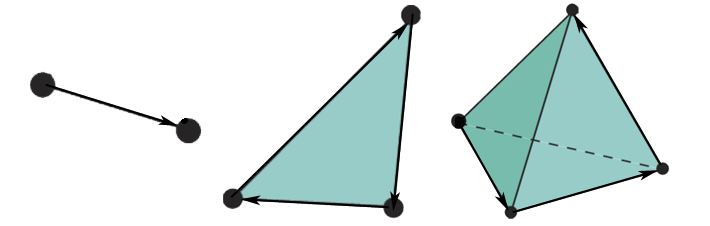}
\caption{\footnotesize{Visual representations of oriented $k$-simplices for $k=1,2,3$, respectively. In particular, note the direction of the arrow indicating the simplex's orientation in each case.\\}}\label{oriented}
\end{figure}

\begin{define}
Let $X$ be a simplicial complex with oriented $k$-simplices $\sigma_{i}$, for $i$ in some index set $I_k$. Fix any arbitrary field $\mathbb{F}$. Define a \textit{k-chain} of $X$ over $\mathbb{F}$ to be a formal sum of the oriented $k$-simplices of $X$, denoted $\sum_{i \in I_k}a_i\sigma_{i}$. Here, $a_i$ is an element of $\mathbb{F}$ for all $i$ in $I_k$. We denote by $C_k(X,\mathbb{F})$, or $C_k(X)$ where context is clear, the set of all $k$-chains of $X$ over $\mathbb{F}$.
\end{define}
\begin{nrem}
Informally-speaking, a $k$-chain can be thought of as assigning elements of $\mathbb{F}$ to the oriented $k$-simplices of $X$.
\end{nrem}
\begin{nrem}
Although the final theoretical results of this paper only require $\mathbb{F}$ to have a ring structure, we will restrict our discussion to fields only. This assumption not only appreciably simplifies the development of persistent homology, but also makes our results visually meaningful in the context of simplicial complexes.\\
\end{nrem}

We now further extend $k$-chains by defining an addition operation $\oplus$ on any two $k$-chains.\\

\begin{define} \label{chainop}
Where $+$ is the addition operation of $\mathbb{F}$, define the binary operator $\oplus$ for \textit{$k$-chain addition} via
\begin{align*}
\oplus:C_k(X,\mathbb{F})\times C_k(X,\mathbb{F}) &\rightarrow C_k(X,\mathbb{F})\\
\sum_{i \in I_k}a_i\sigma_{i} \oplus\sum_{i \in I_k}b_i\sigma_{i}&=\sum_{i\in I_k}(a_i+b_i)\sigma_{i}. \numberthis \label{groupaddition}
\end{align*}\\
\end{define}

As a final result of this subsection, we present a crucial property of the set of $k$-chains of a simplicial complex that ultimately makes persistent homology possible.\\

\begin{theorem} \label{abelian}
Fix a non-negative integer $k$, a simplicial complex $X$, and a field $\mathbb{F}$. The set $C_k(X,\mathbb{F})$ of $k$-chains of $X$ over $\mathbb{F}$, together with the $k$-chain addition $\oplus$, forms an Abelian group.\\
\end{theorem}

\begin{nrem}
This proposition can be proven directly by appealing to the definition of an Abelian group. For brevity, we refer the reader to Appendix Section \ref{appendixabelian} for a rigorous proof, but make note of a few important observations here.\\

The required properties of \textit{closure}, \textit{associativity}, the existence of an \textit{identity element}, and the existence of \textit{inverse elements} for $C_k(X,\mathbb{F})$ under $\oplus$ follow readily from the same properties of the additive operation $+$ in $\mathbb{F}$.  The \textit{commutativity} of $\oplus$ similarly follows from the commutativity of $+$ in $\mathbb{F}$. In particular, let us consider the existence of inverses and a neutral element in $C_k(X,\mathbb{F})$. Where $0$ denotes the additive neutral element of $+$ in $\mathbb{F}$, observe that $C_k(X,\mathbb{F})$ has additive neutral element $\sum_{i\in I_k}0\sigma_i$.\\

Additionally, an element $\sum_{i\in I_k}a_i\sigma_i$ of $C_k(X,\mathbb{F})$ has additive inverse $\sum_{i\in I_k}(-a_i)\sigma_i$, where $-a_i$ is the additive inverse of $a_i$ in $\mathbb{F}$ under $+$. It is here that the necessity of simplex orientation can be seen: for an oriented simplex $\sigma_S$ on some set of points $S$, we say that the inverse of $\sigma_S$ in $C_k(X,\mathbb{F})$, denoted by $-\sigma_S$, is the same simplex $\sigma_S$ but with reverse orientation. This notion connects our intuition with the algebraic structure of $k$-chains in that to undo the ``loop'' implicit in an ordered $k$-simplex, we simply apply the reverse ``loop''---that is, the same simplex but with a reversed orientation.
 
\end{nrem}

\tpoint{Boundary operators\\} \label{boundaryoperatorssection}

In the previous subsection, an Abelian group structure was imposed on the set of $k$-simplices of a simplicial complex. We proceed by examining a map between oriented simplices of different dimension, as well as properties of this map fundamental to persistent homology.\\

For simplicity in this subsection, we suppress $\mathbb{F}$ in all notation outside of formal definitions, and will assume $\mathbb{F}$ to be fixed. Furthermore, we use 1 to represent the multiplicative neutral element of $\mathbb{F}$.\\

\begin{notation}
Denote $[x_0,x_1,\dots,x_k]$ to be the oriented $k$-simplex $\sigma_{\{x_0,x_1,\dots,x_k\}}$ with ordering $(x_0,x_1,\dots,x_k)$.
\end{notation}
\begin{notation}
For $j=0,1,\dots,k$, denote $[x_0,\dots,\hat{x}_j,\dots,x_k]$ to be the same oriented simplex but with $\hat{x}_j$ removed, namely, $[x_0,\dots,x_{j-1},x_{j+1}\dots,x_k]$.\\
\end{notation}

\begin{define} \label{boundary}
Define the \textit{boundary} of the oriented $k$-simplex $[x_0,x_1,\dots,x_k]$ to be
\begin{equation*}
\partial_k[x_0,x_1,\dots,x_k] = \sum_{j=0}^k(-1)^j[x_0,\dots,\hat{x}_j\dots,x_k].
\end{equation*}\\
\end{define}

To better illustrate the purpose and intuition of the boundary of a simplex, we present examples involving general $k$-simplices for $k=0,1,2,3$. See \textbf{Figure \ref{boundaryfig}} for a visual representation of the below examples (except for the trivial case where $k=0$). Let $a$, $b$, $c$, and $d$ be arbitrary points. By Definition \ref{boundary}, observe that
\begin{align}
\partial_0[a]&=0,\label{0bound}\\
\partial_1[a,b]&=[b]-[a],\label{1bound}\\
\partial_2[a,b,c]&=[b,c]-[a,c]+[a,b],\label{2bound}\\
\text{and } \partial_3[a,b,c,d]&=[b,c,d]-[a,c,d]+[a,b,d]-[a,b,c].\label{3bound}
\end{align}

\begin{figure}[!h]
  \centering
    \includegraphics[scale=1]{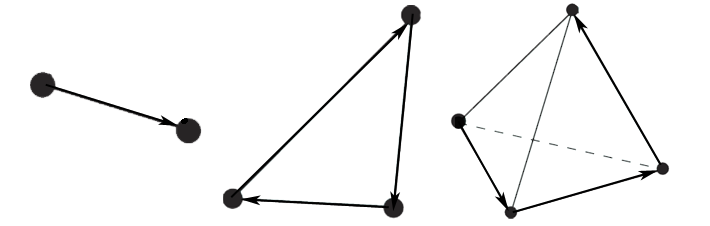}
  \caption{\footnotesize{The boundary of a $k$-simplex for $k=1,2,3$, respectively. Note how the orientation as indicated by the arrows in the latter two diagrams corresponds to the right sides of Equations \ref{2bound} and \ref{3bound}, respectively.\\
Equation \ref{1bound}: a 0-chain remains, the boundary of the original line segment.\label{boundaryfig}\\
Equation \ref{2bound}: a 1-chain remains, forming a loop on the boundary of the original triangular region.\\
Equation \ref{3bound}: a 2-chain remains, forming a loop using the boundary faces of the original tetrahedron.\\}}
\end{figure}

The boundary of a simplex can be developed further through generalisation to the \textit{boundary operator} for $k$-chains, as shown below.\\

\begin{define} \label{boundaryop}
Define the \textit{dimension $k$ boundary operator} on the simplicial complex $X$ via
\begin{align*}
\partial_{k,X}: C_k(X,\mathbb{F}) &\rightarrow C_{k-1}(X,\mathbb{F})\\
\partial_{k,X}\Big(\sum_{i\in I_k}a_i\sigma_i\Big) &= \sum_{i\in I_{k}}a_i\partial_{k-1}(\sigma_i)
\end{align*}
\end{define}
\begin{notation}
As context will make clear whether we are dealing with $\partial_k$ or $\partial_{k,X}$---the boundary of a simplex or of a chain, respectively---we will from now suppress notation and write $\partial_k$ in both cases, for simplicity.\\
\end{notation}

Observe that the boundary operators connect the chain groups of a simplicial complex $X$ by the sequence of maps
\begin{equation} \label{chainmap}
\dots \xrightarrow{\partial_{k+2}}C_{k+1}(X) \xrightarrow{\partial_{k+1}}C_{k}(X) \xrightarrow{\partial_{k}}C_{k-1}(X)\xrightarrow{\partial_{k-1}}\dots \xrightarrow{\partial_{2}}C_{1}(X) \xrightarrow{\partial_{1}}C_{0}(X)\xrightarrow{\partial_0}C_{-1}(X)=\{0\}.
\end{equation}
\begin{nrem}
In order to define $\partial_0$, note that we include $C_{-1}(X)$ as the trivial group $\{0\}$: indeed, the boundary of a 0-simplex is empty, so the boundary operator $\partial_0$ is consistent with the notation introduced thus far.\\
\end{nrem}
As final preparatory work before the formal introduction of simplicial homology groups, we briefly examine properties of the boundary operator and of Equation \ref{chainmap} in the next subsection.\\

\tpoint{Cycles and boundaries\\ \label{cyclesbounds}}

We now focus on properties of the boundary operator previously introduced in Definition \ref{boundaryop}, and in particular, the relationship between the operator's image and kernel. For simplicity, we continue to suppress $\mathbb{F}$ in our notation outside of formal definitions, and assume $\mathbb{F}$ to be fixed.\\

\begin{nlem}\label{homo}
The dimension $k$ boundary operator $\partial_k$ is a homomorphism of groups from $C_k(X)$ to $C_{k-1}(X)$ for all $k\geq1$.\\
\end{nlem}
\begin{proof}
The claim can be proven directly by verifying the definition of a group homomorphism. We will show that the boundary operator $\partial_k$ respects the group operation $\oplus$ of Definition \ref{chainop}. Let $\sum_{i \in I_k}a_i\sigma_i$ and $\sum_{i \in I_k}b_i\sigma_i$ be $k$-chains of a simplicial complex $X$. Observe that\\
\begin{align*}
\partial_k\Big(\sum_{i \in I_k}a_i\sigma_i \oplus\sum_{i \in I_k}b_i\sigma_i\Big) &= \partial_k\Big(\sum_{i \in I_k}(a_i+b_i)\sigma_i\Big) &\text{(by definition of $\oplus$)}\\
&= \sum_{i \in I_k}(a_i+b_i)\partial_k(\sigma_i) &\text{(by definition of $\partial_k$)}\\
&= \sum_{i \in I_k}a_i\partial_k(\sigma_i) \oplus \sum_{i \in I_k}b_i\partial_k(\sigma_i) &\text{(by definition of $\oplus$)}\\
&= \partial_k\Big(\sum_{i \in I_k}a_i\sigma_i\Big) \oplus \partial_k\Big(\sum_{i \in I_k}b_i\sigma_i\Big) &\text{(by definition of $\partial_k$)}\\
\end{align*}
We have shown that the boundary operator $\partial_k$ respects the $C_k(X)$ group operation $\oplus$. Therefore, the boundary operator is a homomorphism of groups.\\
\end{proof}

Recall from group theory that both the \textit{image} and \textit{kernel} of a group homomorphism are themselves groups (\cite{goodman}, Proposition 2.4.12). Then, as a corollary to Lemma \ref{homo}, the image and kernel of the boundary operator are groups, both of which we examine below.\\

\begin{notation}
For non-negative $k$, we denote by $0_k$ the trivial element of the $k$-chain group $C_k(X)$. In other words, we define $0_k=\sum_{i\in I_k}0\sigma_i$. For consistency with Equation \ref{chainmap}, we further denote $0_{-1}$ to be simply 0.\\
\end{notation}

\begin{define} \label{kcycle}
A \textit{$k$-cycle} is a $k$-chain with trivial boundary. More precisely, a $k$-cycle of a simplicial complex $X$ is a $k$-chain $\sum_{i \in I_k}a_i\sigma_{i}$ of $X$ such that
\begin{equation*}
\partial_k\Big(\sum_{i \in I_k}a_i\sigma_{i}\Big)=0_{k-1}.
\end{equation*}
We denote the set of $k$-cycles of a simplicial complex $X$ by $Z_k(X,\mathbb{F})$, or otherwise by $Z_k(X)$ or $Z_k$ where context is clear.\\
\end{define}

It can immediately be seen that the set of $k$-cycles is, by definition, the kernel of the dimension $k$ boundary operator. Therefore, as a corollary to Lemma \ref{homo}, $Z_k(X)$ is a subgroup of $C_k(X)$, for any simplicial complex $X$.\\

\begin{define} \label{kboundary}
A \textit{$k$-boundary} is the boundary of a $(k+1)$-chain. Put precisely, a $k$-chain $\sum_{i \in I_k}b_i\sigma_{i}$ of a simplicial complex $X$ is a $k$-boundary of $X$ if there exists a $(k+1)$-chain $\sum_{i \in I_{k+1}}a_i\sigma_{i}$ in $C_{k+1}(X)$ such that
\begin{equation*}
\partial_{k+1}\Big(\sum_{i \in I_{k+1}}a_i\sigma_{i}\Big)=\sum_{i \in I_k}b_i\sigma_{i}.
\end{equation*}
We will denote the set of $k$-boundaries of $X$ by $B_k(X,\mathbb{F})$ or, where context is clear, simply by $B_k(X)$ or $B_k$.\\
\end{define}

Once again, we see immediately that the set of $k$-boundaries is, by definition, the image of the dimension $(k+1)$ boundary operator. Therefore, as a corollary to Lemma \ref{homo}, $B_k(X)$ is a subgroup of $C_k(X)$ for any simplicial complex $X$.\\

The rest of this subsection will prove a relationship between the set of $k$-cycles and $k$-boundaries.

\begin{nlem} \label{boundofbound}
For any $k\geq0$, the image of a $k$-boundary under the dimension $k$ boundary operator is the trivial $(k-1)$-chain. Equivalently, for any integer $k\geq1$,
\begin{equation*}
\partial_{k-1}\partial_{k}\Big(\sum_{i \in I_{k+1}}a_i\sigma_i\Big)=0_{k-2}.
\end{equation*}\\
\end{nlem}

We will again refer the reader to Appendix Section \ref{appendixbound} for a rigorous proof of this claim, and instead give an example below.\\

Consider the general oriented 3-simplex $[a,b,c,d]$. The image of this simplex under the dimension 3 boundary operator is
\begin{align*}
\partial_3[a,b,c,d] &= [b,c,d]-[a,c,d]+[a,b,d]-[a,b,c].
\end{align*}
Now applying the dimension 2 boundary operator to this result, we observe that
\begin{align*}
\partial_2\big(\partial_3[a,b,c,d]\big) &= \partial_2[b,c,d] -\partial_2[a,c,d] + \partial_2[a,b,d] - \partial_2[a,b,c]\\
&=\big([c,d]-[b,d]+[b,c]\big) - \big([c,d]-[a,d]+[a,c]\big) + \big([b,d]-[a,d]+[a,b]\big) - \big([b,c]-[a,c]+[a,b]\big)\\
&= 0_1,
\end{align*}
Informally put, the boundary of a boundary is trivial.\\

\begin{theorem}
For a simplicial complex $X$, every $k$-boundary of X is a $k$-cycle of $X$.\\
\end{theorem}

\begin{proof}
Using the previously-introduced notation, observe that the given statement is equivalent to $B_k(X)\subset Z_k(X)$. We prove this below.\\

Let $\sum_{i \in I_k}b_i\sigma_i$ be any element of $B_k(X)$. By definition of a $k$-boundary, there exists some $(k+1)$-chain $\sum_{i \in I_{k+1}}a_i\sigma_i$ such that
\begin{equation}
\partial_{k+1}\Big(\sum_{i \in I_{k+1}}a_i\sigma_i\Big)=\sum_{i \in I_k}b_i\sigma_i. \label{subsetpf}
\end{equation}\\

Applying $\partial_k$ to the chosen $k$-boundary, observe that
\begin{align*}
\partial_{k}\Big(\sum_{i \in I_k}b_i\sigma_i\Big) &= \partial_{k}\partial_{k+1}\Big(\sum_{i \in I_{k+1}}a_i\sigma_i\Big) &\text{(by Equation \ref{subsetpf})}\\
&= 0_{k-1}.&\text{(by Lemma \ref{boundofbound})}\\
\end{align*}

Therefore, by definition of a kernel, we have that $\sum_{i \in I_k}b_i\sigma_i$ is an element of $Z_k(X)$, the kernel of $\partial_k$.\\

Since this $k$-chain was chosen arbitrarily from $B_k(X)$, it follows that $B_k(X)\subset Z_k(X)$.\\
\end{proof}
\begin{nrem}
Since we have thus far proven that $B_k(X)$ is itself both a group and a subset of $Z_k(X)$, it follows that $B_k(X)$ is a subgroup of $Z_k(X)$.
\end{nrem}

\section{Simplicial Homology} \label{simplicialhomologysection}

In the previous section, we added a notion of orientation to the simplicial complex introduced in Section \ref{basicsimplex} and defined a group structure on the set of $k$-chains of a simplicial complex. In particular, the results of Section \ref{boundaryoperatorssection} on cycles and boundaries will be fundamental in the development of homology groups in the context of simplicial complexes, called \textit{simplicial homology}.\\

We continue, in this section, to suppress mention of the arbitrary field $\mathbb{F}$ in our notation outside of formal definitions wherever possible. Furthermore, we continue to assume that such an $\mathbb{F}$ is fixed.\\

\tpoint{Homology groups\\ \label{homologygroupssection}}

In this subsection, we fix a simplicial complex $X$ and again denote the set of $k$-chains, $k$-cycles, and $k$-boundaries of $X$ by $C_k$, $Z_k$ and $B_k$, respectively. Furthermore, we use a standard notation to denote operations on a set: where $g$ is an element of a group $G$ with group operation $+$, and where $H$ is a subset of $G$, we define
\begin{equation*}
g+H = \{g+h \mid h\in H\}.\\
\end{equation*}\\

For brevity, we refer the reader to Appendix Section \ref{appendixnormal} for the definition of a normal subgroup and proofs for the related results presented below.\\

We have shown previously in Theorem \ref{abelian} that $C_k$ is an Abelian group. Therefore, it follows immediately that every subgroup of $C_k$ is a normal, Abelian subgroup of $C_k$ (see Appendix Section \ref{appendixnormal}). Furthermore, we proved in Section \ref{cyclesbounds} that both $Z_k$ and $B_k$ are subgroups of $C_k$, and that $B_k$ is a subgroup of $Z_k$.\\

An immediate consequence of these results is a key prerequisite for the development of persistent homology, namely that $B_k$ is a normal subgroup of $Z_k$ (again by Appendix Section \ref{appendixnormal}, since $Z_k$ is Abelian). Consequently, we can now define the \textit{homology group} of a simplicial complex.\\

\begin{define} \label{homologygroup}
The \textit{$k^{\text{th}}$ homology group} $H_k(X,\mathbb{F})$ of a simplicial complex $X$ is the collection of unique \textit{cosets} of $B_k(X,\mathbb{F})$ in $Z_k(X,\mathbb{F})$---that is, the unique equivalence classes of form 
\begin{equation*}
z+B_k(X,\mathbb{F}),
\end{equation*}
where $z$ is a $k$-cycle in $Z_k(X,\mathbb{F})$.\\

Equivalently, we write
\begin{equation*}
H_k(X,\mathbb{F}) = \frac{Z_k(X,\mathbb{F})}{B_k(X,\mathbb{F})}
\end{equation*}\\
to mean that $H_k(X,\mathbb{F})$ is the quotient group of $Z_k(X,\mathbb{F})$ modulo $B_k(X,\mathbb{F})$.\\
\end{define}

\begin{notation}
Where context is clear, we will denote $H_k(X,\mathbb{F})$ by $H_k(X)$ or by $H_k$, for simplicity.\\
\end{notation}
\begin{theorem} \label{vectorspace}
The $k^\text{th}$ homology group $H_k(X,\mathbb{F})$ is a vector space over $\mathbb{F}$.\\
\end{theorem}
\begin{nrem}
For brevity, a proof of this claim is omitted from this paper, though the result follows readily once an appropriate vector addition and $\mathbb{F}$-scalar multiplication on $H_k(X,\mathbb{F})$ is defined. In fact, all points of the vector space criteria follow immediately from the status of $C_k(X,\mathbb{F})$ as an Abelian group and $\mathbb{F}$ as a field.\\

The definitions of the above-mentioned vector operations are fairly intuitive and not widely used in this paper, so we refer the reader to Appendix Section \ref{vectorops} for full details.\\
\end{nrem}

Intuitively, the elements of $H_k$ describe the different ``kinds'' of cycles present in a simplicial complex without regard to the complex's boundary elements. This general intuitive understanding of a homology group is formalised in the above definition with the distinct sets of the form $z+B_k$, called cosets. For a fixed $k$-cycle $z_0$, the coset $z_0+B_k$ contains all $k$-cycles of the simplicial complex that differ from $z_0$ only by $k$-boundaries. Thus, if another $k$-cycle $y_0$ differs from $z_0$ by only $k$-boundaries, then $y_0$ is also a member of the coset $z_0+B_k$.\\

We again emphasize that the elements of $H_k$ are the \textit{distinct} equivalence classes of $k$-cycles in the given simplicial complex. Put more rigorously, two $k$-cycles $z_0$ and $y_0$ are in the same coset, or class, if $z_0\oplus(-y_0)$ is a $k$-boundary (where $-y_0$ is the inverse of $y_0$ in $C_k$, as defined in the proof of Proposition \ref{abelian}). In other words, the difference between $z_0$ and $y_0$ is composed only of $k$-boundaries.\\

\tpoint{Homology group example\\} \label{egsection}

To illustrate the concepts developed up to this point, we present a simple yet informative example where we explicitly calculate two homology groups of a given simplicial complex \cite{homoeg}. In this subsection, we use standard notation to represent the kernel and image of a function $f$, namely, $\ker f$ and $\Ima f$, respectively. We also employ the intuitive vector addition and scalar multiplication operations defined for the $k^\text{th}$ homology group, as presented in Appendix Section \ref{vectorops}. Lastly, we use familiar notation from linear algebra to denote the space spanned by a set of chains, as set out below.\\

\begin{notation}
Let $\sigma_i$ be $k$-chains and $a_i$ elements of some fixed field $\mathbb{F}$, for $i=1,2,\dots,n$. Denote by $\spn_\mathbb{F}\{\sigma_1,\sigma_2,\dots,\sigma_n\}$ the set of all $k$-chains of the form $a_1\cdot\sigma_1\oplus a_2\cdot\sigma_2\oplus \dots\oplus a_n\sigma_n$. We refer to $\sigma_1,\sigma_2,\dots,\sigma_n$ as \textit{generators} of $\spn_\mathbb{F}\{\sigma_1,\sigma_2,\dots,\sigma_n\}$.\\
\end{notation}

Considered below is the the simplicial complex $X$, as presented in \textbf{Figure \ref{eg}}.
\begin{figure}[h]
  \centering
    \includegraphics[scale=1]{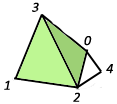}
  \caption{\footnotesize{A simplicial complex with vertices 0,1, 2, 3, and 4 as shown. \label{eg}}}
\end{figure}

We begin by listing the $k$-chain groups of $X$ by taking the span of all $k$-simplices in $X$, for $k=0,1,2$.
\begin{align}
C_0(X) &= \spn\{[0],[1],[2],[3],[4]\}\\
C_1(X) &= \spn\{[0,1],[0,2],[0,3],[0,4],[1,2],[1,3],[2,3],[2,4]\}\\
C_2(X) &= \spn\{[0,1,2],[0,1,3],[0,2,3],[1,2,3]\}
\end{align}\\

Let us first consider the dimension 0 homology group $H_0(X)$. By Definition \ref{homologygroup} of a homology group and by Section \ref{cyclesbounds}, we know that\\
\begin{equation} \label{egH0}
H_0(X) = \frac{Z_0(X)}{B_0(X)} = \frac{\ker \partial_0}{\Ima \partial_1}\\
\end{equation}\\

We will first calculate $\ker\partial_0$. Observe that, for every $0$-simplex [x], we have $\partial_0[x]=0$. Therefore, every 0-simplex is mapped to 0 under the boundary map, and so the set of all 0-simplices is in the kernel of $\partial_0$. In other words,
\begin{equation*}
C_0(X)\subset\ker\partial_0.\\
\end{equation*}
The kernel of $\partial_k$ is necessarily a subset of $C_0(X)$ by the definition of a kernel, so the reverse inclusion also holds. Therefore, we conclude that
\begin{equation} \label{z0}
\ker\partial_0=C_0(X).\\
\end{equation}\\

Let us now move on to determine $\Ima\partial_1$. We can calculate the generators of $\Ima\partial_1$ as image of the generators of $C_0(X)$. In other words,

\begin{align*}
B_0(X) = \Ima\partial_1 &= \partial_1\spn\{[0,1],[0,2],[0,3],[0,4],[1,2],[1,3],[2,3],[2,4]\}\\
&= \spn\{\partial_1[0,1],\, \partial_1[0,2],\, \partial_1[0,3],\, \partial_1[0,4],\, \partial_1[1,2],\, \partial_1[1,3],\, \partial_1[2,3],\, \partial_1[2,4]\}\\
&= \spn\{[1]-[0],\, [2]-[0],\, [3]-[0],\, [4]-[0],\, [2]-[1],\, [3]-[1],\, [3]-[2],\, [4]-[2]\}. \numberthis \label{b0}
\end{align*}\\

Therefore, by substituting Equations \ref{z0} and \ref{b0} into Equation \ref{egH0}, we obtain\\
\begin{equation*}
H_0(X) = \frac{\spn\{[0],[1],[2],[3],[4]\}}{\spn\{[1]-[0],\, [2]-[0],\, [3]-[0],\, [4]-[0],\, [2]-[1],\, [3]-[1],\, [3]-[2],\, [4]-[2]\}}.
\end{equation*}\\

Recall that any two $k$-chains $z_0$ and $y_0$ are considered equivalent in the $k^\text{th}$ homology group if their difference $z_0\oplus(-y_0)$ is a $k$-boundary. In this particular example, one can show that the difference of any two elements of $Z_0$ is a $0$-boundary.\\

For example, consider the 0-cycles $[3]$ and $[4]$, and the $0$-boundaries $[3]-[2]$ and $[4]-[2]$. Note that all of these are generators in the previous equation. Observe that
\begin{align*}
\big([3]-[2]\big) \oplus -\big([4]-[2]\big) &= \big([3]-[2]\big) \oplus \big(-[4]+[2]\big)\\
&= [3] + (1-1)[2] - [4]\\
&= [3] - [4]
\end{align*}
Note that $\big([3]-[2]\big) \oplus -\big([4]-[2]\big)$ is a $0$-boundary since $B_0(X)$ is a group and is closed under 0-chain addition. Therefore, we see that [3] and [4] are equivalent in $H_0(X)$, as the difference between these two cycles is an element of the boundary group.\\

As stated above, this same result holds true in general for every generator of the 0-cycles $Z_0(X)$. From this, it follows that every 0-cycle of $X$ is equivalent to every other 0-cycle of $X$. In other words, $[0]$, $[1]$, $[2]$, $[3]$, and $[4]$ are all elements of the same coset, namely, $[0]+B_0(X)$. Of course, since these cycles are equivalent, we may also represent this coset as $[x]+B_0(X)$, where $x$ is any of 0, 1, 2, 3, or 4. Finally, since $[0]$ is not a 0-boundary itself, we note that this coset is non-trivial---that is, $[0]+B_0(X)$ is not $B_0(X)$.\\

As a result, there is exactly one non-trivial element of the $0^\text{th}$ homology group $H_0(X)$, namely $[0]+B_0(X)$. Pictorially, this result corresponds to the fact that $X$, as shown in \textbf{Figure \ref{eg}}, is composed of exactly one connected component.\\

The $1^\text{st}$ homology group $H_1(X)$ can be calculated in a similar way. Once again, by the definition of a homology group and by previous results, we have that\\
\begin{equation} \label{egH1}
H_1(X) = \frac{Z_1(X)}{B_1(X)} = \frac{\ker\partial_1}{\Ima\partial_2}
\end{equation}\\

Though the work required is somewhat tedious without any additional techniques or methods, one can calculate $\ker\partial_1$ by direct computation to show that
\begin{align*}
\ker\partial_1 = \spn\{&[0,1]+[0,3]-[1,3],\,[0,2]+[2,3]-[0,3],\,[1,2]+[2,3]-[1,3],\,[0,1]+[1,2]-[0,2],\\&[0,2]+[2,4]-[0,4]\} \numberthis \label{z1}
\end{align*}\\

Notably less strenuous is the calculation for $B_1(X)$:
\begin{align*}
B_1(X) &= \Ima\partial_2 = \partial_2\spn\{[0,1,2],[0,1,3],[0,2,3],[1,2,3]\}\numberthis \label{b1}\\
&= \spn\{[0,1]+[1,2]-[0,2],\,[0,1]+[0,3]-[1,3],\,[0,2]+[2,3]-[0,3],\,[1,2]+[2,3]-[1,3]\}.
\end{align*}\\

Now substituting Equations \ref{z1} and \ref{b1} into Equation \ref{egH1}, we obtain, for $H_1(X)$,\\
\begin{align*}
\frac{\spn\{[0,1]+[0,3]-[1,3],[0,2]+[2,3]-[0,3],[1,2]+[2,3]-[1,3],[0,1]+[1,2]-[0,2],[0,2]+[2,4]-[0,4]\}}{\spn\{[0,1]+[1,2]-[0,2],\,[0,1]+[0,3]-[1,3],\,[0,2]+[2,3]-[0,3],\,[1,2]+[2,3]-[1,3]\}}
\end{align*}\\

Observe that the first four generators of $Z_1(X)$ in the numerator are also generators of $B_1(X)$ in the denominator: therefore, any combination of these four elements will necessarily be a boundary! On the other hand, one can see that the fifth generator of $\ker\partial_1$, namely, $[0,2]+[2,4]-[0,4]$ is independent of the boundary elements---this is intuitively clear since none of the boundary generators concern the point labelled as 4.\\

We conclude that $H_1(X)$ has two elements: besides the trivial class $B_1(X)$, we also have the non-trivial class $\big([0,2]+[2,4]-[0,4]\big)+B_1(X)$.\\

Recall from Section \ref{boundaryoperatorssection} that 1-cycles can be visualised as loops, and consider the above result in the context of \textbf{Figure \ref{eg}}. The trivial coset of $H_1(X)$ can be thought of as the class of all loops on \textbf{Figure \ref{eg}} that can be shrunk down to a single point: the loops of this class are exactly those loops that do not make use of the ``arm'' formed by the vertices labelled 2, 4, and 0. On the other hand, the non-trivial coset $\big([0,2]+[2,4]-[0,4]\big)+B_1(X)$ corresponds exactly to those loops that use this extra ``arm'' and hence cannot be reduced to a single point. Refer to \textbf{Figure \ref{egloop}} for a visual example of both cases.
\begin{figure}[h]
  \centering
    \includegraphics[scale=1]{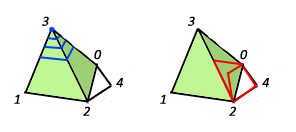}
  \caption{\footnotesize{Loops on the simplicial complex of \textbf{Figure \ref{eg}}. The blue loop on the left can be shrunk to a single point, as shown. On the other hand, the red loop on the right cannot be similarly reduced. These loops are examples of elements in the trivial and non-trivial cosets of $H_1(X)$, respectively.}}\label{egloop}
\end{figure}

In short, the result that $H_1(X)$ has exactly one non-trivial element corresponds to the fact that $X$ has exactly one 1-dimensional ``hole''---namely, the loop formed by the vertices labeled 2, 0, and 4.\\

\tpoint{Betti numbers\\ \label{bettinumbers}}

As a final result of this section, we develop a numerical summary of any homology group, called the \textit{Betti number}. The existence of this numerical descriptor stems from the result of Lemma \ref{vectorspace} that $H_k(X,\mathbb{F})$ is itself a vector space over the field $\mathbb{F}$. Consequently, the homology group $H_k(X,\mathbb{F})$ has a well-defined \textit{dimension}---that is, intuitively, the number of non-trivial elements that can be used to generate the homology group.\\

\begin{define}
The \textit{$k^\text{th}$ Betti number} of a simplicial complex $X$, denoted $\beta_k(X)$, is the dimension of the $k^\text{th}$ homology group $H_k(X,\mathbb{F})$ as a vector space over fixed field $\mathbb{F}$. We write
\begin{equation*}
\beta_k(X) = \dim H_k(X,\mathbb{F}),
\end{equation*}
and note in particular that the Betti number $\beta_k(X)$ is independent of the choice $\mathbb{F}$.\\
\end{define}

The Betti numbers of a given simplicial complex $X$ provide an easily-interpretable description of the topology of $X$. As demonstrated by the example in Section \ref{egsection}, the dimension of the $k^\text{th}$-homology group---or equivalently, the number of non-trivial generators---reveals how many holes of dimension $k$ are present in $X$. It is crucial to note that each non-trivial element of a homology group corresponds to a topological feature of $X$, as explained below.\\

In this sense, $\beta_0(X)$ can be interpreted as the number of connected components of $X$, and $\beta_1(X)$ as the number of loops of $X$---or equivalently, the number of 2-dimensional regions enclosed by $X$. Furthermore, $\beta_2$ is the number of \textit{voids}, or enclosed 3-dimensional regions, of $X$.\\

The Betti number can be defined similarly outside the context of simplicial homology. For example, the $k^\textit{th}$ Betti number of a compact manifold $M$, denoted $\beta_k(M)$, is the dimension of the $k^\textit{th}$ homology group of $M$. In other words, $\beta_k(M)$ can be interpreted as the number of connected components, loops, voids, and so on, of the manifold $M$. For intuition, we give a couple examples below of the first three Betti numbers for the sphere and torus.\\

Consider the standard sphere $S$ in three-dimensional Euclidean space. Observe that, since $S$ has a single connected component, $\beta_0(S)=1$. As all loops on $S$ are trivial---that is, since all loops can be shrunk to a single point---we have $\beta_1(S)=0$. Furthermore, since $S$ encloses a single three-dimensional region, we have $\beta_2(S)=1$.\\

As another example, consider the standard torus $T$ in three-dimensional Euclidean space. The torus is composed of a single connected component, so $\beta_0(T)=1$. Also, $T$ has two non-trivial classes of loops, namely, loops around the central ``hole'' of the torus and loops around ``tube'' of the torus: we then have $\beta_1(T)=2$. See \textbf{Figure \ref{torusloops}} for a visualisation of these two classes of loops. Finally, since $T$ encloses one three-dimensional space inside its ``tube'', we have $\beta_2(T)=1$.\\

\begin{figure}[h]
  \centering
    \includegraphics[scale=0.2]{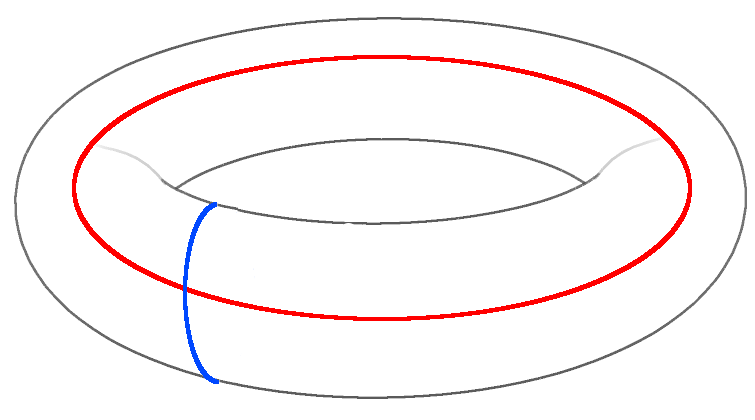}
  \caption{\footnotesize{Examples of the two non-trivial classes of loops on a torus, displayed in red and blue \cite{torusloops}.\label{torusloops}}}
\end{figure}

\tpoint{Persistence barcodes\label{persistencebarcodes}\\}

As a final theoretical topic, we will join simplicial homology with the filtrations of simplicial complexes introduced in Section \ref{pcconstructions}. We will continue to use the Vietoris-Rips complex construction of Definition \ref{VR-Complex} for its computational efficiency. Assumed throughout this subsection is the notation of Section \ref{pcconstructions} regarding simplices and Vietoris-Rips simplicial complexes.\\

We begin with a basic property of the a Vietoris-Rips filtration for a fixed point cloud $P$.\\

\begin{nprop}
For two fixed, real, and non-negative $\varepsilon$ and $\varepsilon'$ with $\varepsilon\leq\varepsilon'$, the Vietoris-Rips complex $V_\varepsilon(P)$ is nested inside $V_{\varepsilon'}(P)$. That is, every simplex of $V_\varepsilon(P)$ is also a simplex of $V_{\varepsilon'}(P)$.\\
\end{nprop}

\begin{proof}
The proof of this claim follows immediately from Definition \ref{VR-Complex} of a Vietoris-Rips complex.\\
Suppose $\sigma_S$ is a simplex of $V_\varepsilon(P)$ for some set of points $S$ in $P$. Then by definition, for every pair of points $x$ and $y$ in $S$, it follows that
\begin{equation*}
d(x,y)\leq\varepsilon.
\end{equation*}
By assumption, $\varepsilon\leq\varepsilon'$, and so also,
\begin{equation*}
d(x,y)\leq\varepsilon'.
\end{equation*}\\

Therefore, by definition of the Vietoris-Rips complex, it follows that $\sigma_S$ is a simplex of $V_{\varepsilon'}(P)$. We have then proven the desired result, namely that $V_\varepsilon(P)$ is a subset of $V_{\varepsilon'}(P)$.\\
\end{proof}

We can form a \textit{chain of nested simplicial complexes} using the Vietoris-Rips construction by varying $\varepsilon$. Indeed, given a sequence of increasing $\varepsilon_i$, where $\varepsilon_1<\varepsilon_2<\dots\varepsilon_n<\varepsilon_{n+1}<\dots$ we have\\
\begin{equation} \label{vrstream}
V_{\varepsilon_1}(P)\subset V_{\varepsilon_2}(P)\subset\dots\subset V_{\varepsilon_n}(P)\subset V_{\varepsilon_{n+1}}(P) \subset\dots,
\end{equation}\\

As a result, we can define maps between the homology groups of these complexes, namely\\

\begin{equation} \label{homologychain}
H_k\big(V_{\varepsilon_1}(P)\big)\xrightarrow{\varphi_1} H_k\big(V_{\varepsilon_2}(P)\big)\xrightarrow{\varphi_2} \dots \xrightarrow{\varphi_{n-1}} H_k\big(V_{\varepsilon_n}(P)\big)\xrightarrow{\varphi_n} H_k\big(V_{\varepsilon_{n+1}}(P)\big)\rightarrow\dots,
\end{equation}\\

where $k\geq2$. For the rest of this subsection, we will denote $H_k\big(V_{\varepsilon_n}(P)\big)$ by $H_k^{\varepsilon_n}$, and the map from $H_k^{\varepsilon_n}$ to $H_k^{\varepsilon_{n+1}}$ by $\varphi_n$, for ease of notation. While the specific maps $\varphi_n$ are not significant for our purposes, the effect of these maps on the above homology groups is certainly of note. Consider the map $\varphi_n$, and suppose that $z_0$ and $y_0$ are distinct elements of $H_k^{\varepsilon_{n}}$. Note that elements of $H_k^{\varepsilon_n}$ are always mapped forward to $H_k^{\varepsilon_{n+1}}$ in Equation \ref{homologychain} by $\varphi_n$. However, $\varphi_n$ is not necessarily surjective---that is, there may be some elements of $H_k^{\varepsilon_{n+1}}$ that are not the image of any element of $H_k^{\varepsilon_{n}}$ under $\varphi_n$. Such elements are said to be \textit{born at time} $\varepsilon_n$.\\

Observe that $\varphi_n$ will map $z_0$ and $y_0$ to elements of $H_k^{\varepsilon_{n+1}}$ that may either be distinct or identical. Where $\varphi_n(z_0)$ and $\varphi_n(y_0)$ are distinct, we say that $z_0$ has \textit{persisted} from $\varepsilon_n$ to $\varepsilon_{n+1}$, and similarly so for $y_0$. On the other hand, if $\varphi_n(z_0)$ is equal to $\varphi_n(y_0)$, we say that one of $z_0$ or $y_0$ has \textit{died}. By convention that will soon become apparent, we choose the element of the pair that was born last to be the one to die at $\varepsilon_{n+1}$. Thus, if $z_0$ was born before $y_0$, we say that $y_0$ \textit{dies at time} $\varepsilon_{n+1}$.\\

At this point, we can develop a simple yet intuitive visual representation of the birth and death times of all homological elements appearing in Equation \ref{homologychain}.\\

\begin{define} \label{barcodedef}
Fix some $k\geq0$, and let $Y_k$ be the set of all $k^{\text{th}}$ homology group elements of Equation \ref{homologychain} at the time they are born. Consider, for each element $y$ in $Y_k$, the real, half-open interval $[b_y,d_y)$, where $b_y$ and $d_y$ are the birth and death times of $y$, respectively. Define the \textit{$k$-barcode} corresponding to the filtration given in Equation \ref{vrstream} by the collection of intervals
\begin{equation*}
\{[b_y,d_y)\mid y\in Y_k\}
\end{equation*}\\
\end{define}

Although beyond the scope of this paper, it can be shown that the homological features appearing in a nested filtration of complexes---for example, the object given in Equation \ref{homologychain}---is ismorphic to the collection of the filtration's corresponding $k$-barcodes for all $k\geq0$. In other words, the birth and death times of all features in a filtration uniquely determines the filtration's corresponding barcode, and vice-versa. We can now, easily and without algebraic notation, represent the homology of a filtration of simplicial complexes as an intuitive series of intervals!


\section{Components of the Proposed Clustering Analysis \label{sectioncomponents}\\}

Now that the theory of persistent homology and its application to simplicial complex filtrations has been developed in the previous sections, we begin the second major portion of this paper. Our focus now shifts to the application of persistent homology to real-world data---in particular, we examine a sizable dataset from the field of linguistics.\\

This section serves as an introduction to the major components relevant to the data analysis that was carried out. In the following subsections, we briefly discuss the \textit{Edinburgh Associative Thesaurus} dataset, the \textit{modularity index} for assessing the quality of a clustering method, the \textit{Markov Clustering} algorithm, and the particular persistent homology techniques applied in this study.\\

\tpoint{Edinburgh Associative Thesaurus\\ \label{eatsubsection}}

The Edinburgh Associative Thesaurus (EAT) is a large dataset containing information on mental associations made between words of the English language \cite{eat}. As discussed in Section \ref{linguisticapplications}, an individual will associate various ideas, concepts, and notions with a given word. As expected, these associations will vary from person to person based on culture, personal experience, worldview, or any number of factors that shape how an individual thinks. For example, one person may associate the word GERMAN most strongly with the word FRENCH, whereas another may associate GERMAN with KRAUT, as shown in \textbf{Figure \ref{french}}. The set of word associations form, for each person, a network between words of the English language, called a \textit{word association network}.\\

\begin{figure}[h]
  \centering
    \includegraphics[scale=0.15]{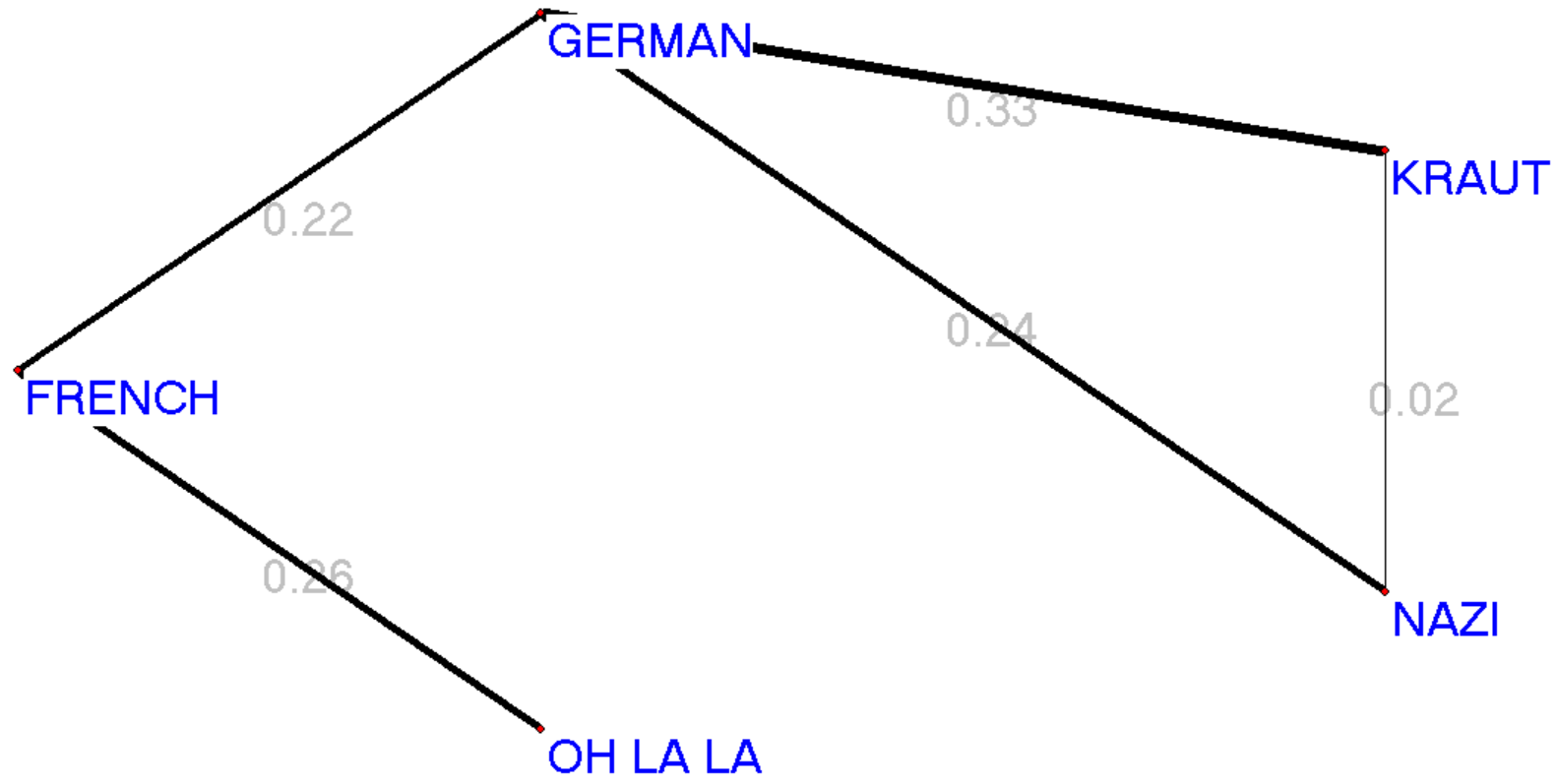}
  \caption{\footnotesize{A small portion of a word association network, highlighting relationships made with the words FRENCH and GERMAN. In this diagram, edge thickness is proportional to the association strength, also labeled numerically on each edge.}\label{french}}
\end{figure}

Differences and similarities in word association networks between persons are of particular interest to researchers, again as discussed in Section \ref{linguisticapplications}. The EAT, for example, has previously been used to find and classify semantic and psychological links between words, as well as to maximize advertising efficacy by making use of common associations.\\

The EAT database was constructed using 8,400 \textit{stimulus words}. Each of these stimulus words was presented on paper to approximately 100 different subjects. Each subject was prompted to write down, as quickly as possible, the first word that came to mind after viewing the stimulus.\\

The data comprising the EAT contains all stimulus words, all responses, and the number of times that a response was given for each stimulus. Based on this data, we assigned a numerical index to each word and calculated the proportion of occurrence for each ordered pair of words. For example, if 25 out of 100 people presented with the word CAT responded with DOG, then the proportion of occurrence of DOG after seeing CAT is 0.25. As the proportion of occurrence is not necessarily symmetric, we took the \textit{association strength} between two words to be the maximum proportion of occurrence between them. Continuing the previous example, if the proportion of occurrence of CAT after seeing DOG is 0.40, we take the strength of the association between CAT and DOG to be 0.4---that is, the maximum of 0.25 and 0.4.\\

In total, the data used in our analysis included 305,134 associations between 23,181 unique words.\\

\tpoint{Modularity index\label{modularitysection}\\}

As one of this project's main objectives is to compare the clustering abilities of persistent homology to other methods, it is essential that we have a scale to measure the performance of each technique. To this end, we make use of the \textit{modularity index} $Q$ for weighted graphs \cite{modularity}---that is, graphs for which a weight has been assigned to each edge. As the formal development of modularity is outside the scope of this paper, we instead provide only a definition and intuitive explanation of this measure of clustering performance.\\

The modularity index $Q$ is a numerical value between -1 and 1 that describes how well a given graph has been partitioned into clusters. Loosely-speaking, modularity measures the difference in density between the connections within clusters and the connections between clusters. A set of clusters that more effectively separates the vertices of a graph will have a higher modularity: we thus seek to maximize $Q$ to obtain the best possible clustering of a graph.\\

In the below definition of modularity, we assume the following notation. Let the vertices of a graph be labelled according to some index set $V$, and let $\omega_{i,j}$ represent the weight of the edge between the vertices labelled $i$ and $j$. Note that $\omega_{i,j}$ is taken to be 0 if no such edge exists. Let $M$ be the sum of all edge weights in the graph, and $k_i$ the sum of the weights of all edges attached to the vertex labelled $i$. Lastly, define $\delta(i,j)$ to be the function that equals 1 when the vertices labelled $i$ and $j$ have been assigned to the same cluster, and 0 otherwise.\\

\begin{define}
Using the notation above, given a weighted graph and a partitioning of its vertices into clusters, define the modularity index $Q$ of this clustering to be

\begin{equation*}
Q=\frac{1}{M}\sum_{i,j\in V}\big[\omega_{i,j}-\frac{k_ik_j}{M}\big]\delta(i,j)
\end{equation*}
\end{define}

Essentially, the sum above only considers pairs of vertices $i$ and $j$ in the same cluster. Intuitively, we then see that $Q$ is increased by $\omega_{i,j}$, the weight of edges within a cluster, and decreased by $\frac{k_ik_j}{M}$, a measure of the complexity of the graph around the vertices $i$ or $j$. In the context of machine learning, this definition of modularity is essentially the fundamental problem of balancing the interpretability of a model with its complexity.\\

\tpoint{Markov Clustering algorithm\\}

\textit{Markov Clustering} (MCL) is an algorithm developed by van Dongen for separating a graph or network into clusters---that is, partitioning the vertices of a graph into non-overlapping subsets containing vertices that are similar in some way \cite{mcl}. In general, a cluster of a graph is characterised by a higher proportion of edges within the cluster then outside the cluster. As discussed in section \ref{linguisticapplications}, the problem of clustering is relevant to the fields of image analysis, machine learning, general pattern recognition in computing science, and bioinformatics.\\

Most clustering methods and algorithms, however, become computationally infeasible for increasingly large datasets---that is, these algorithms are not scalable. The MCL algorithm, on the other hand, is presented as a computationally efficient and scalable means of extracting clusters from even very large networks. Although MCL has been used widely in the field of bioinformatics, the algorithm has been applied previously to linguistics in the creation of a dictionary of French synonyms (\cite{ling1}, Section 4.1) and a study of word clustering in the Japanese language (\cite{ling1}, Section 4.2).\\

Based on its precedent use in linguistic analysis, we chose to use MCL in this study as a performance benchmark for persistent homology. Although a detailed exposition of MCL is outside the scope of this paper, we present here a brief, intuitive description of this algorithm.\\

The scalability of MCL to large graphs stems from the algorithm's use of random walk simulations on the graph being considered. More specifically, the idea underlying MCL is that, by randomly traveling along the edges of a graph, one is more likely at any point to stay within a single cluster than one is to exit the cluster. Based on this idea, MCL alternates between periods of simulating long and short random walks. These periods are respectively referred to as the \textit{expansion} and \textit{inflation stages} of the algorithm. Longer random walks are more likely to travel between clusters, thus allowing potential clusters to expand and include more vertices. Shorter random walks, on the other hand, are more likely to stay within a cluster---this serves to remove weak elements of a potential cluster and strengthen the connection between vertices that are strongly similar.

In our application of MCL, the words of the EAT are interpreted as vertices of a graph. Furthermore, for any pair of words, the weight of the edge connecting them is taken to be their association strength.\\

MCL is dependent on a choice of \textit{inflation parameter} that determines the ``strength'' of the inflation stage of the algorithm. In our study, we performed MCL on the EAT database for a wide range of inflation parameter values. We then calculated the modularity $Q$ of the clustering created by each iteration of the algorithm. The results of this test can be found in Section \ref{mclresults}.\\

\tpoint{Clustering with complexes and persistent homology\label{phtests}\\}

The main focus of this study is persistent homology's ability to find clusters and higher-dimensional topological features such as loops and voids in large datasets. In this subsection, we describe specifically how persistent homology was used to extract clusters and other topological features from the EAT. We will assume the notation used in Section \ref{pcconstructions} pertaining to complex constructions.\\

In order to apply the Vietoris-Rips construction to the EAT, we generalise Definition \ref{VR-Complex}. Since our dataset $P$---that is, the words in the EAT---cannot be placed meaningfully in Euclidean space, we instead take the metric $d(x,y)$ to be one minus the association strength between the two words $x$ and $y$. In other words, $d$ becomes a measure of dissimilarity. Note that this transformation of association strength to dissimilarity is necessary to ensure that pairs of words with high association strength have a short ``distance'' between them. After this modification, though we do not have a true metric and cannot properly visualise the complexes created, we are still able to construct a Vietoris-Rips filtration.\\

To illustrate, consider the example previously given in Section \ref{eatsubsection} using the words CAT and DOG. We will denote, as above, the words in the EAT dataset by $P$. We previously supposed the association strength between CAT and DOG to be 0.4: therefore, the ``distance'', or dissimilarity, between these words is
\begin{equation*}
d(\text{CAT},\text{DOG})=1-0.4=0.6.
\end{equation*}
Therefore, the 1-simplex built from the ``points'' CAT and DOG will be present in the complex $V_{0.7}(P)$, but not in $V_{0.5}(P)$.\\

Using the R-TDA package for topological data analysis, we constructed a filtration of Vietoris-Rips complexes on the EAT and determined the $k$-barcodes, for $k=0,1,2$, corresponding to this filtration. Recall from Section \ref{persistencebarcodes} that the barcode is equivalent to the persistent homology of the filtration: both contain information about the birth and death times of each topological feature that appears.\\

We first considered maximizing cluster modularity over the set of all Vietoris-Rips complexes created. In subsequent sections, we refer to this method as \textit{simple clustering by similarity}. Ultimately, this is a naive method that clusters together pairs of words with similarity above a specified threshold parameter---in other words, we simply take connected components as they appear in a single Vietoris-Rips complex. Note that this method only uses properties of the Vietoris-Rips complex, and not of persistent homology. We performed such clustering over a large number of threshold values in order to maximize the cluster modularity $Q$.\\

In contrast, the next method, referred to as \textit{clustering by persistence}, does make use of the persistent homology of the constructed Vietoris-Rips filtration. Here we consider the \textit{persistence}---that is, the difference between the birth and death time---of each 0-dimensional homological feature. We fix a persistence threshold parameter and cluster two words $x$ and $y$ together if and only if the 1-simplex connecting $x$ and $y$ has a lifetime greater than the specified threshold.\\

We note that, at the time of this study, the R-TDA package did not have the functionalities necessary to perform the above analysis. Although we developed code to extract clusters and topological features of arbitrary dimension from a homology of the Vietoris-Rips filtration, such code will be made available in a future publication and is not presented here.\\


\section{Clustering Results\\ \label{resultssection}}

This section presents results for each of the clustering tests introduced in Section \ref{sectioncomponents} and compares these methods using the modularity index as defined in Section \ref{modularitysection}. Furthermore, we discuss a modification of the clustering by persistence method to increase persistent homology's clustering effectiveness relative to MCL.\\

\tpoint{Markov Clustering results \label{mclresults}\\}

We applied the Markov Clustering algorithm to partition the 23,181 words of the EAT dataset into groups of closely-associated words. Due to the dependence of MCL on a choice of inflation parameter, we iterated the algorithm 241 times using a range of inflation parameter values between 1.20 and 6.00. For each iteration, we calculated the modularity of the clustering produced and looked to maximize this quantity over the inflation parameter values tested. \textbf{Figure \ref{mcl_Q}} presents a plot of the modularity value calculated for each MCL iteration against the inflation parameter used.\\

\begin{figure}[h]
  \centering
    \includegraphics[scale=0.6]{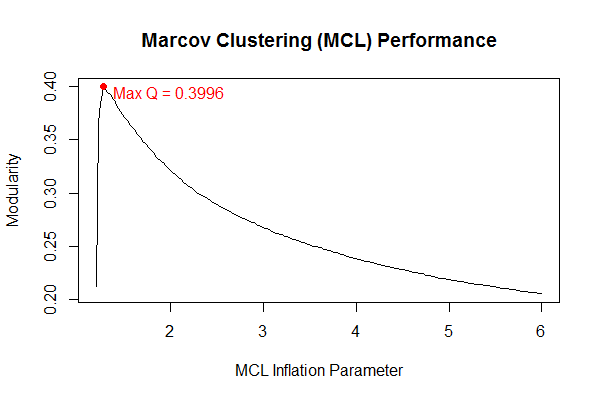}
  \caption{\footnotesize{A plot showing the relationship between MCL's inflation parameter value and the quality of the clusters produced by the algorithm for the EAT dataset.}\label{mcl_Q}}
\end{figure}

As shown in \textbf{Figure \ref{mcl_Q}}, the maximum modularity value attained by MCL is 0.3996, occurring at an inflation parameter value of 1.28. This particular iteration produced 319 distinct word clusters.\\

\tpoint{Simple clustering by similarity results\\}

In this test, we used individual Vietoris-Rips complexes to cluster the words of the EAT dataset, as set out in Section \ref{phtests}. Similar to the MCL algorithm, the Vietoris-Rips complex construction is dependent on a parameter $\varepsilon$, adhering the notation of \ref{pcconstructions}. As such, we looked to maximize modularity over the 31 values of $\varepsilon$ chosen. A plot of the results is presented in \textbf{Figure \ref{ph_time_Q}}.\\

\begin{figure}[h]
  \centering
    \includegraphics[scale=0.6]{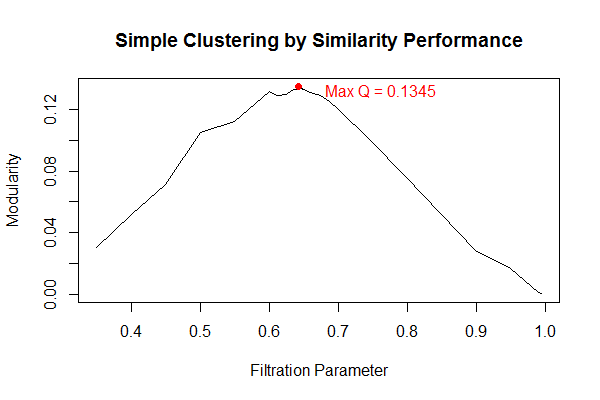}
  \caption{\footnotesize{A plot showing the relationship between Vietoris-Rips parameter $\varepsilon$, called the \textit{filtration parameter}, and modularity index in the simple clustering by similarity method.}\label{ph_time_Q}}
\end{figure}

Modularity was at a maximum of 0.1345 for a value 0.6421 of the filtration parameter $\varepsilon$. Furthermore, the number of clusters produced at optimal modularity was found to be 21,523. Observe that the optimal modularity value produced by this method is considerably lower than that of the MCL algorithm.\\

\tpoint{Clustering by persistence results\label{phresults}\\}

In our final test, we used the persistent homology of a filtration of Vietoris-Rips complexes to cluster the words of the EAT. As described in Section \ref{phtests}, our method of clustering was once again dependent on a \textit{persistence threshold} parameter. Recall that, for this method, we cluster two words together if and only if the 0-simplex connecting the two words has a lifetime greater than the chosen persistence threshold. As in previous subsections, we maximized modularity over 33 threshold values: a plot of the results is presented in \textbf{Figure \ref{ph_pers_Q}}.\\

\begin{figure}[h]
  \centering
    \includegraphics[scale=0.6]{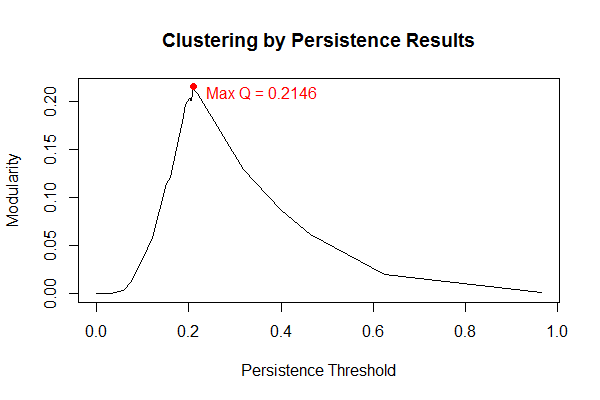}
  \caption{\footnotesize{A plot showing the relationship between persistence threshold and modularity index in the clustering by persistence method.}\label{ph_pers_Q}}
\end{figure}

Observe that the maximum modularity value attained by this method is 0.2146, occurring at a persistence threshold of 0.209. At this point, 18,882 distinct clusters were present. Although this maximum modularity value is higher than that of the simple clustering by similarity method, it is still lower than the maximum modularity attained by the MCL algorithm.\\

\tpoint{Discussion\\}

In this subsection, we compare the results of each of the previous methods, and focus specifically on simple clustering by similarity and clustering by persistence. Discussed is a modification that may potentially increase persistent homology's efficacy as a clustering method, relative to MCL.\\

As noted in Section \ref{phtests}, the simple clustering by similarity method is naive in that it only considers the connected components of a single simplicial complex in a Vietoris-Rips filtration. This results in a clustering that groups any given vertex according to solely the strength of its direct connections to other vertices. Furthermore, the modularity values for this method suffer from a problem that largely motivates persistent homology in the first place---noise in the data. Short-lived features, or in this setting, connected components that are created but quickly join with a larger component in the Vietoris-Rips filtration, are counted the same as components with a long lifetime. The failure of the simple clustering method to address this issue fragments the generated clusters into smaller pieces and yields a lower overall modularity.\\

For this reason, clustering by persistence performs notably better than simple clustering by similarity. The former considers the overall significance of each connected component in the entire filtration and removes those components with lifetimes deemed insignificant. However, the method still suffers from a similar weakness in that it doesn't look beyond a vertex's direct connection to its neighbours, ultimately fragmenting potential clusters into small pieces. This effect is clearly seen in the high number of clusters generated by these two methods, as shown in previous subsections. MCL seems to do better in this regard by considering the probability that a random walk on the vertices of a graph stays within a given cluster, rather than considering just each individual connection between vertices.\\

For these reasons, a modification to persistent homology that would likely improve the method's clustering ability for the EAT dataset is to construct simplicial complexes based on some measure of \textit{vertex density} rather than just the similarity between two words. One such (simplistic) density measure could, for example, consider the vertices adjacent to the immediate neighbours of a given vertex. Loosely-speaking, the proportion of those vertices that are still immediately adjacent to the given vertex would function as a measure of how dense a dataset is at the given vertex. The use of such a method would require a different formulation of persistent homology, so we do not discuss it further in this paper.\\


\section{Linguistic Interpretation of Results\\} \label{interpretationsection}

In this final section, we present some specific results of our analysis of the EAT dataset with persistent homology. In particular, we give examples of clusters, loops, and voids, as well as potential interpretations of each of these features in a linguistic context.\\

To extract the clusters displayed in this section, we disregarded all edges with a lifetime less than approximately 0.209---the value of the persistence threshold maximizing the modularity index, as found in Section \ref{phresults}. During our analysis, we noted that higher-dimensional features such as loops and voids have significantly shorter lifespans than do clusters: as such, we selected the loops and voids presented here from those found to be most persistent---that is, those having the longest lifetime in the Vietoris-Rips filtration. Features were then extracted from the R-TDA results using the code described in Section \ref{phtests}.\\

Each of the images displayed in this section were created using the Pajek program for large network analysis and visualisation \cite{pajek}. We make particular note that the placement of each word in a diagram has no effect on the interpretation of the network structures formed. Rather, we used the Kamada-Kawai and Fruchterman-Reingold graph-drawing algorithms \cite{graphdraw} to arrange the words in an visually-pleasing and interpretable way. In each diagram, the thickness of an edge connecting two words is proportional to the association strength of the word pair, also indicated numerically.\\

Similar images of the features found in the EAT are presented in Appendix Section \ref{appendiximages}.

\tpoint{Clusters\\}

Clusters are the simplest and most interpretable feature in the EAT data. Our results agree with common intuition in that words tend to associate closely with others sharing a similar underlying idea. However, we observed it was not uncommon for the words of a given cluster to be connected to a few central, yet thematically unrelated words. This sort of structure appears in \textbf{Figure \ref{cluster_pupil}} with the words EYE and SCHOOL, and in \textbf{Figure \ref{cluster_death}} with DEATH.\\

While most clusters seem to be based largely on a common theme, a number of clusters span multiple ideas and concepts. Of particular interest are these intra-cluster connections between distinct ideas. The cluster in \textbf{Figure \ref{cluster_pupil}} is a representative example, where two clearly distinct themes are present: eyesight and school. Connecting these two ideas is PUPIL, a word which may be taken as both a synonym to the word student or an anatomical part of the eye. We observe that, in this case, the formation of a cluster linking these two concepts is due to PUPIL's multiple meanings in the English language.\\

\begin{figure}[h]
  \centering
    \includegraphics[scale=0.2]{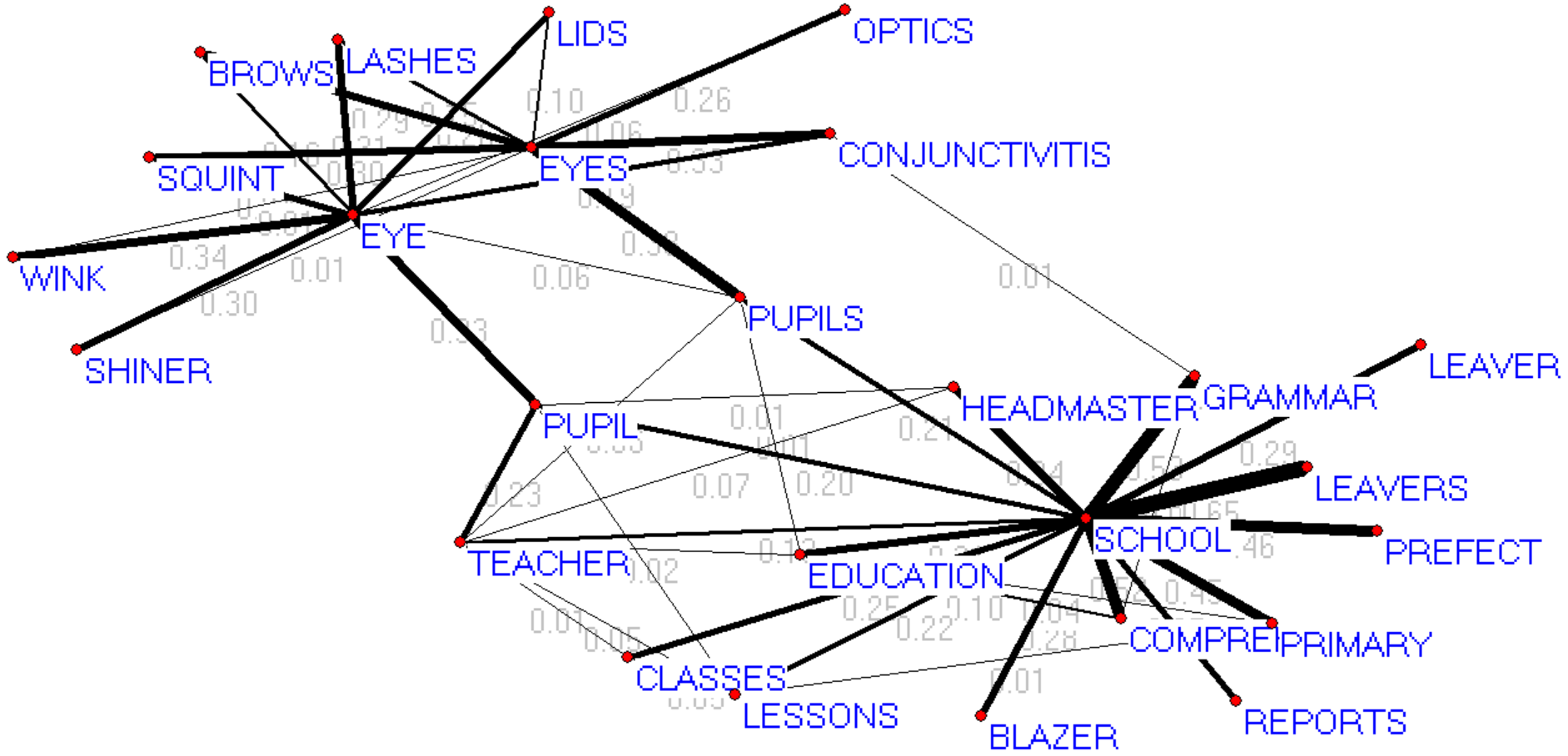}
  \caption{\footnotesize{A word cluster displaying two separate themes: eyesight and school. The connection between them is mediated by PUPIL, due to this word's multiple meanings in the English language.}\label{cluster_pupil}}
\end{figure}

\textbf{Figure \ref{cluster_death}} displays the same phenomenon. In contrast to \textbf{Figure \ref{cluster_pupil}}, however, the connection between the two themes of thought and death through CONCEPTION is not due entirely to the structure of the English language. First, CONCEPTION is related to IDEA, perhaps through a mental association with the word CONCEPT or the way in which ideas are ``born'' in an individual's mind. Second, CONCEPTION may be taken in a biological sense as a synonym to BIRTH.\\

\begin{figure}[h]
  \centering
    \includegraphics[scale=0.2]{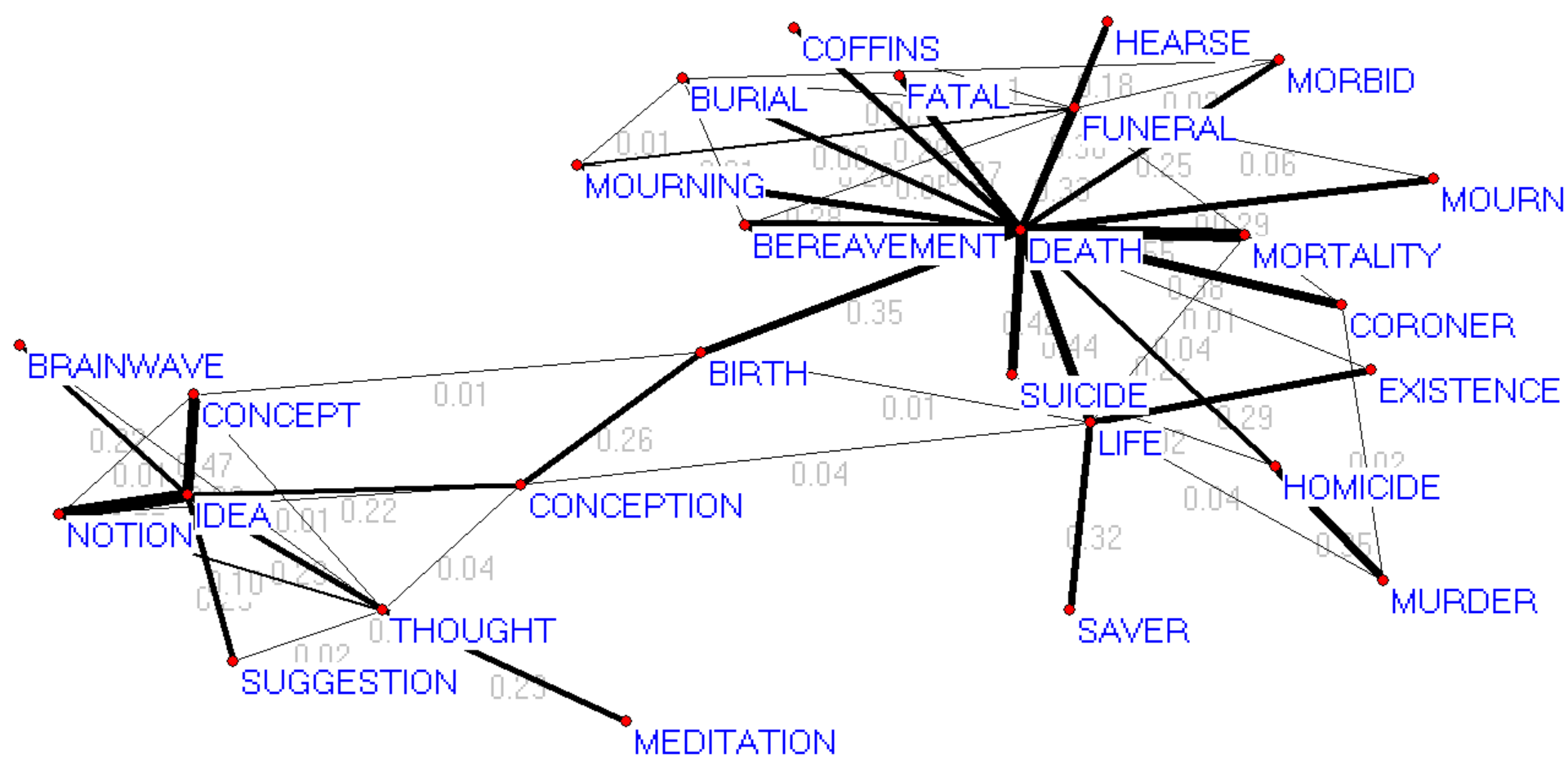}
  \caption{\footnotesize{A cluster whose members group around DEATH and IDEA. The word CONCEPTION connects the themes present in the graph.}\label{cluster_death}}
\end{figure}

\tpoint{Loops\\}

The example loop presented in this subsection may be easily seen by following word associations of higher strength---that is, the thicker edges---in the diagram. Thinner edges represent weak associations that were disregarded in our clustering method, as described at the beginning of this section.\\

In a linguistic context, a loop in the EAT data can be interpreted as a chain of closely-associated words linking one word back to itself. For example, in \textbf{Figure \ref{loop_cat}}, a certain ``train of thought'' connecting the words of the loop can be seen. We make particular note that, as in clusters, the words of a loop need not share a common theme.\\

\begin{figure}[h]
  \centering
    \includegraphics[scale=0.25]{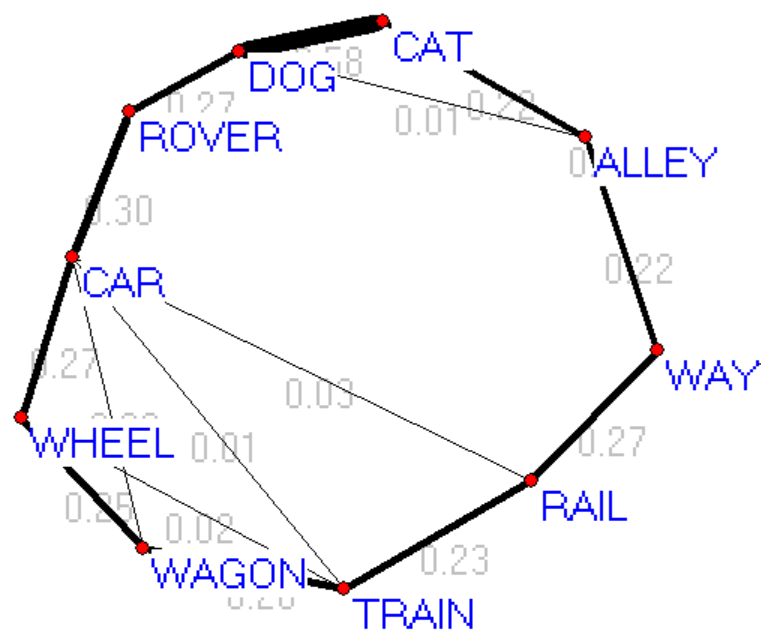}
  \caption{\footnotesize{A small loop found in the EAT data. The sequence of words composing the loop form a ``train of thought'' between them. For example, we see that CAR and DOG are connected by ROVER, both a kind of vehicle and a common dog name.}\label{loop_cat}}
\end{figure}

\tpoint {Voids\\}

Lastly, we present an example of a void found in the EAT data, and compare these features to loops. Voids are topologically equivalent to spheres: in general, they enclose some three-dimensional space. In contrast, loops, as shown in the previous section, enclose two-dimensional spaces.\\

Both loops and voids are composed of links between strongly-associated words. A loop, however, is restricted in the sense that, at any particular word, there are only two directions in which the loop can proceed. A void, on the other hand, due to its higher-dimensional nature, is not restricted in this way. This ``freedom'' is illustrated in \textbf{Figure \ref{void_math}}, particularly in how EQUATION is linked to all of ALGEBRA, MATHS, MATHEMATICS, and SUM---all of the other words of the void, in this case. For this reason, the words of a void are more closely associated with each other than the words of a loop.\\

We found that the words of a void generally share exactly one common theme. In \textbf{Figure \ref{void_math}}, this theme is clearly mathematics. In contrast, the loop of \textbf{Figure \ref{loop_cat}} contained words related to both animals and transportation.\\

\begin{figure}[h]
  \centering
    \includegraphics[scale=0.2]{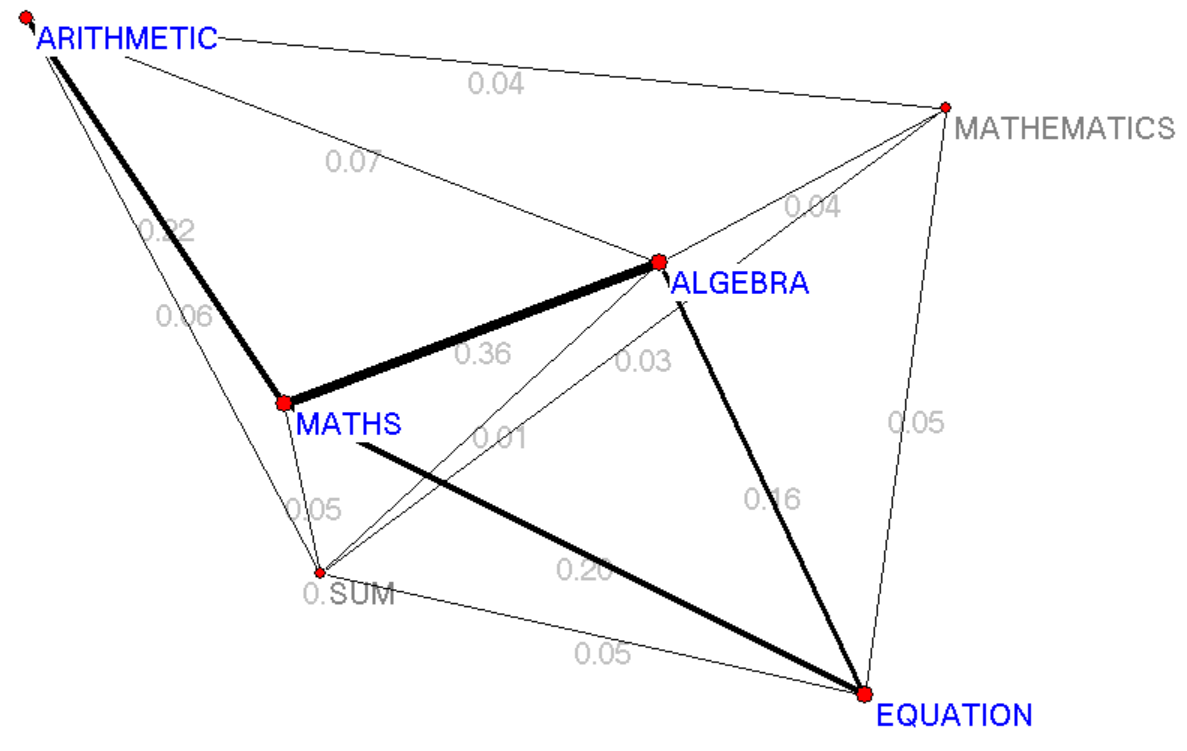}
  \caption{\footnotesize{A visual representation of a void found in the EAT data. Note in particular the high level of connectivity between words of the void and their central theme of mathematics. As the words here were visualised in three-dimensional space, we indicate vertices further in the background with grey labels.}\label{void_math}}
\end{figure}

\newpage
\section{Appendix: Algebraic Proofs and Definitions \label{appendixproofs}}

This appendix contains supplemental materials and proofs deemed too technical for the main body of Sections \ref{boundscyclessection} and \ref{simplicialhomologysection} of this paper. We assume the notations introduced in those sections.

\tpoint{Proof of Proposition \ref{abelian}\\} \label{appendixabelian}

\textbf{Proposition.} \textit{Fix a non-negative integer $k$, a simplicial complex $X$, and a field $\mathbb{F}$. The set $C_k(X,\mathbb{F})$ of $k$-chains of $X$ over $\mathbb{F}$, together with simplex addition $\oplus$, forms an Abelian group.}\\

\begin{proof}
We will verify the standard group criteria and show that $\oplus$ is commutative. Let $\underset{i\in I_k}{\sum}a_i\sigma_i$ and $\underset{i\in I_k}{\sum}b_i\sigma_i$ be arbitrary elements of $C_k(X,\mathbb{F})$.\\

\textit{$\bullet$ Closure under $\oplus$:}\\

Observe that
\begin{equation*}
\underset{i\in I_k}{\sum}a_i\sigma_i \oplus \underset{i\in I_k}{\sum}b_i\sigma_i = \underset{i\in I_k}{\sum}(a_i+b_i)\sigma_i.
\end{equation*}
Since $\mathbb{F}$ is closed under $+$ as a field, it follows that $(a_i+b_i)$ is in $\mathbb{F}$ for all $i$ in $I_k$. Therefore, $\underset{i\in I_k}{\sum}(a_i+b_i)\sigma_i$ is an element of $C_k(X,\mathbb{F})$, proving that $C_k(X,\mathbb{F})$ is closed under $\oplus$.\\

\textit{$\bullet$ Commutativity of $\oplus$:}\\

Observe that
\begin{align*}
\underset{i\in I_k}{\sum}a_i\sigma_i \oplus \underset{i\in I_k}{\sum}b_i\sigma_i &= \underset{i\in I_k}{\sum}(a_i+b_i)\sigma_i &\text{(by definition of $\oplus$)}\\
&=\underset{i\in I_k}{\sum}(b_i+a_i)\sigma_i &\text{(by commutativity of $+$ in $\mathbb{F}$)} \\
&=\underset{i\in I_k}{\sum}b_i\sigma_i \oplus \underset{i\in I_k}{\sum}a_i\sigma_i &\text{(by definition of $\oplus$)}
\end{align*}
Therefore $\oplus$ is commutative in $C_k(X,\mathbb{F})$.\\

\textit{$\bullet$ Identity element:}\\

Let 0 denote the neutral additive element of $\mathbb{F}$ under $+$. We will prove that $\underset{i\in I_k}{\sum}0\sigma_i$ is additive neutral for $\oplus$ in $C_k(X,\mathbb{F})$. Observe that\\
\begin{align*}
\underset{i\in I_k}{\sum}a_i\sigma_i \oplus \underset{i\in I_k}{\sum}0\sigma_i = \underset{i\in I_k}{\sum}(a_i+0)\sigma_i = \underset{i\in I_k}{\sum}a_i\sigma_i
\end{align*}
since 0 is additive neutral in $\mathbb{F}$. Additionally, by the commutativity of $\oplus$ proven above,
\begin{equation*}
\underset{i\in I_k}{\sum}a_i\sigma_i \oplus \underset{i\in I_k}{\sum}0\sigma_i = \underset{i\in I_k}{\sum}0\sigma_i \oplus \underset{i\in I_k}{\sum}a_i\sigma_i = \underset{i\in I_k}{\sum}a_i\sigma_i.
\end{equation*}

Therefore, $\underset{i\in I_k}{\sum}0\sigma_i$ is additive neutral for $\oplus$ in $C_k(X,\mathbb{F})$.\\

\textit{$\bullet$ Inverse elements:}\\

Given any $\underset{i\in I_k}{\sum}a_i\sigma_i$ in $C_k(X,\mathbb{F})$, consider $\underset{i\in I_k}{\sum}(-a_i)\sigma_i$, where $(-a_i)$ denotes the additive inverse of $a_i$ in $\mathbb{F}$. Observe that
\begin{align*}
\underset{i\in I_k}{\sum}a_i\sigma_i \oplus \underset{i\in I_k}{\sum}(-a_i)\sigma_i &= \underset{i\in I_k}{\sum}(a_i+-a_i)\sigma_i=\underset{i\in I_k}{\sum}0\sigma_i
\end{align*}
since $-a_i$ is the additive inverse of $a_i$ under $+$.\\

Therefore, $\underset{i\in I_k}{\sum}(-a_i)\sigma_i$ is the additive inverse of $\underset{i\in I_k}{\sum}a_i\sigma_i$ under $\oplus$.\\\\

Therefore, by the above criteria, $C_k(X,\mathbb{F})$ with group operation $\oplus$ forms an Abelian group.\\
\end{proof}

\tpoint{Proof of Lemma \ref{boundofbound}} \label{appendixbound}\\

\textit{Lemma.} For any $k\geq0$, the image of a $k$-boundary under the dimension $k$ boundary operator is the trivial $(k-1)$-chain. Equivalently, for any integer $k\geq1$
\begin{equation*}
\partial_{k-1}\partial_{k}\Big(\sum_{i \in I_{k}}a_i\sigma_i\Big)=0_{k-2}.
\end{equation*}\\

\begin{proof}
Fix $k\geq1$ as supposed, and let $[x_0,x_1,\dots,x_k]$ be any oriented $k$-simplex. Observe that
\begin{align*}
\partial_{k-1}\partial_{k}[x_0,x_1,..,x_k] &= \partial_{k-1}\sum_{i=0}^k(-1)^i[x_0,..,\hat{x}_i,..,x_k]\\
&= \sum_{i=0}^k(-1)^i\partial_{k-1}[x_0,..,\hat{x}_i,..,x_k]\\
&= \sum_{i=0}^k\Big[\sum_{j=0}^{i-1}(-1)^j(-1)^i[x_0,..,\hat{x}_j,..,\hat{x}_i,..,x_k] + \sum_{j=i+1}^{k}(-1)^{j-1}(-1)^i[x_0,..,\hat{x}_i,..,\hat{x}_j,..x_k]\Big]\\
&= \sum_{i=0}^k\sum_{j=0}^{i-1}(-1)^j(-1)^i[x_0,..,\hat{x}_j,..,\hat{x}_i,..,x_k] + \sum_{i=0}^k\sum_{j=i+1}^{k}(-1)^{j-1}(-1)^i[x_0,..,\hat{x}_i,..,\hat{x}_j,..x_k]\\
&= \sum_{0\leq j<i\leq k}(-1)^j(-1)^i[x_0,..,\hat{x}_j,..,\hat{x}_i,..,x_k] - \sum_{0\leq i< j\leq k}(-1)^{j}(-1)^i[x_0,..,\hat{x}_i,..,\hat{x}_j,..x_k]\\
&= \sum_{0\leq j<i\leq k}(-1)^j(-1)^i[x_0,..,\hat{x}_j,..,\hat{x}_i,..,x_k] - \sum_{0\leq j< i\leq k}(-1)^{i}(-1)^j[x_0,..,\hat{x}_j,..,\hat{x}_i,..x_k]\\
&= 0_{k-2}.\\
\end{align*}

Therefore, for all $k\geq1$, we have $\partial_{k-1}\partial_{k}[x_0,x_1,..,x_k]=0$. Extending this result to $k$-chains, we have
\begin{equation*}
\partial_{k-1}\partial_{k}\Big(\sum_{i \in I_{k}}a_i\sigma_i\Big)=\sum_{i \in I_{k}}a_i\partial_{k-1}\partial_{k}\sigma_i= 0_{k-2},
\end{equation*}
Thus proving the desired result.\\
\end{proof}

\tpoint{Normal subgroups and related results \label{appendixnormal}\\}

In this subsection, we recall the definition of a \textit{normal subgroup} and further prove some basic results required in Section \ref{homologygroupssection}.\\

Assume throughout that $G$ is a group, where the group operation on elements $g_1$ and $g_2$ of $G$ is denoted by $g_1g_2$. Furthermore, let the inverse of $g_1$ under the group operation be denoted by $g_1^{-1}$. Lastly, we suppose that $H$ is a subgroup of $G$.\\

\textit{Definition.} We say that $H$ is a normal subgroup of $G$ if, for every element $g$ of $G$ and every element $h$ of $H$, the element $ghg^{-1}$ is in $H$. Equivalently, we say that $H$ is invariant under conjugation by $G$.\\

\textit{Lemma.} If $G$ is an Abelian group, then every subgroup $H$ of $G$ is a normal subgroup of $G$.\\

\begin{proof}
Suppose $g$ and $h$ are arbitrary elements of $G$ and $H$, and suppose that $G$ is an Abelian group. We will verify that $H$ is a normal subgroup of $G$ by appealing to the definition of a normal subgroup. Observe that
\begin{align*}
ghg^{-1}&=gg^{-1}h &\text{(since $G$ is Abelian)}\\
&=h\in H
\end{align*}
Therefore, $ghg^{-1}$ is in $H$, and it follows by definition that $H$ is a normal subgroup of $G$.
\end{proof}

\tpoint{Vector operations on homology groups \label{vectorops}}\\

In this section, we define a vector addition $\oplus_H$ and scalar multiplication $\cdot$ appropriate for the homology group $H_k(X,\mathbb{F})$ as a vector space.\\

\begin{define}
Let $\oplus_H$ be a binary operation on $H_k(X,\mathbb{F})$ defined via
\begin{align*}
\oplus_H: H_k(X,\mathbb{F})\times H_k(X,\mathbb{F})&\rightarrow H_k(X,\mathbb{F})\\
\big(z_0+B_k(X)\big) \oplus_H \big(y_0+B_k(X)\big) &= (z_0\oplus y_0)+B_k(X),
\end{align*}
where we recall that $\oplus$ is the additive operation of the Abelian group $C_k(X,\mathbb{F})$ presented in Definition \ref{chainop}.\\
\end{define}
\begin{notation}
We will suppress notation and write $\oplus_H$ as $\oplus$, as context will make clear whether we are dealing with the sum of homology group elements or of chains.\\
\end{notation}
\begin{nrem}
The above addition operation follows the standard group theoretic formulation for the addition of cosets.\\
\end{nrem}

\begin{define}
Let $\cdot$ be the $\mathbb{F}$-scalar multiplication defined via\\
\begin{align*}
\cdot: \mathbb{F}\times H_k(X,\mathbb{F})&\rightarrow H_k(X,\mathbb{F})\\
c\cdot\big(\sum_{i \in I_k}a_i\sigma_i + B_k(X)\big)&= \sum_{i \in I_k}(ca_i)\sigma_i+B_k(X),
\end{align*}
where $ca_i$ denotes the result of the multiplicative operation of $\mathbb{F}$ on $c$ and $a_i$.\\
\end{define}

\section{Appendix: Additional Images \label{appendiximages}}

\begin{figure}[h]
  \centering
    \includegraphics[scale=0.25]{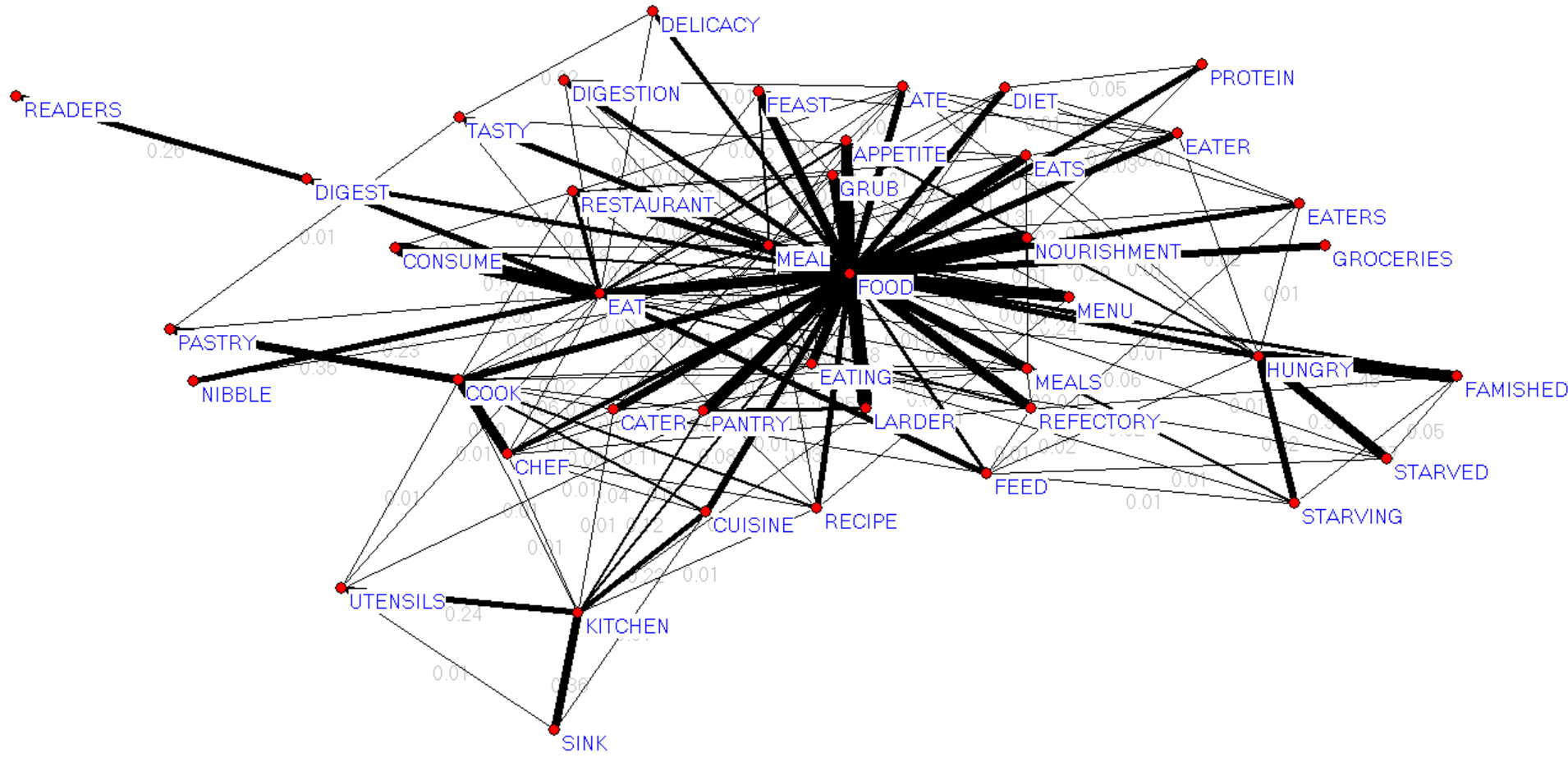}
  \caption{\footnotesize{Highly-centralised clustering around FOOD.}}
\end{figure}

\begin{figure}[h]
  \centering
    \includegraphics[scale=0.25]{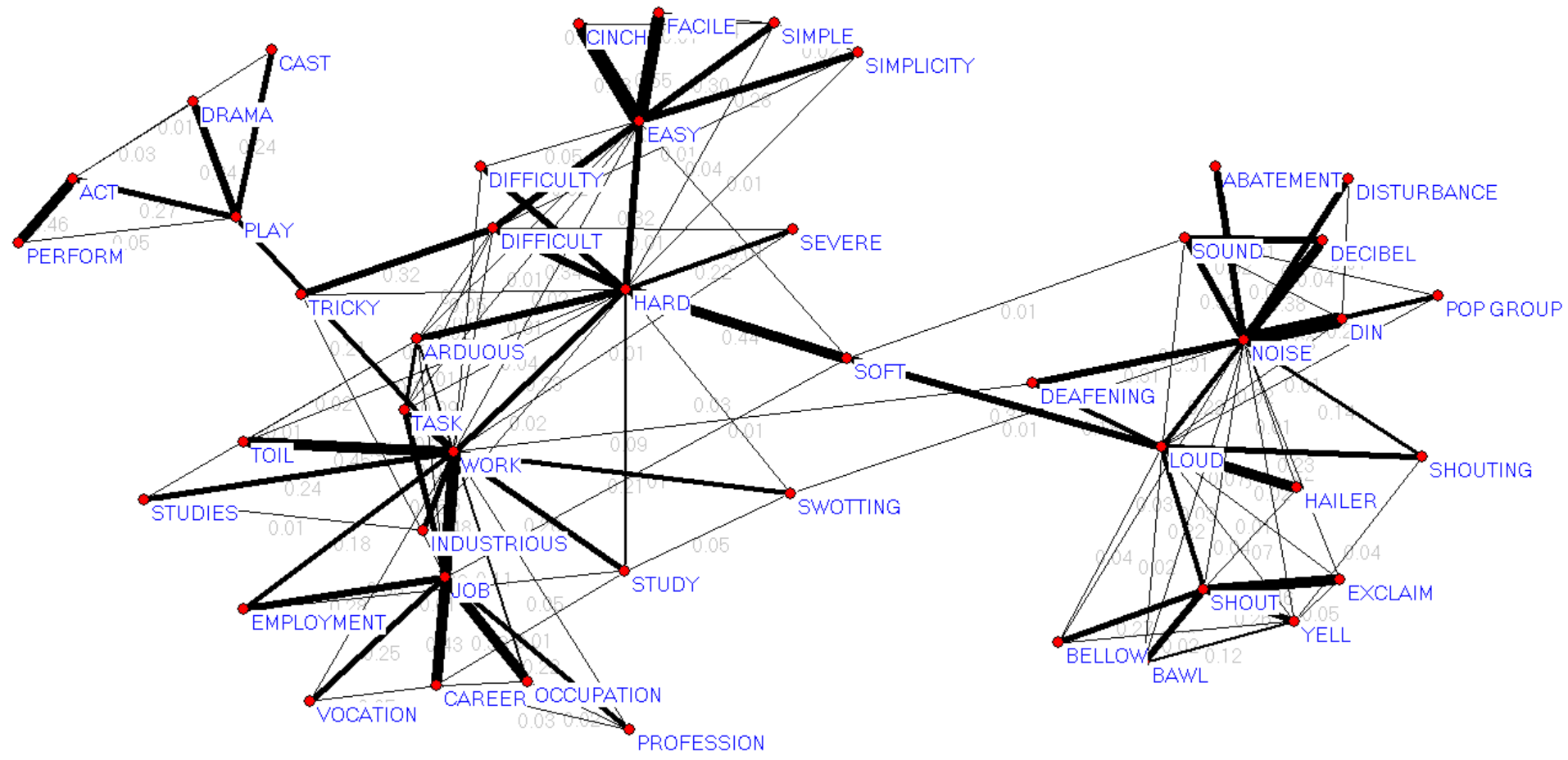}
  \caption{\footnotesize{A number of themes present in a single cluster.}}
\end{figure}

\begin{figure}[h]
  \centering
    \includegraphics[scale=0.25]{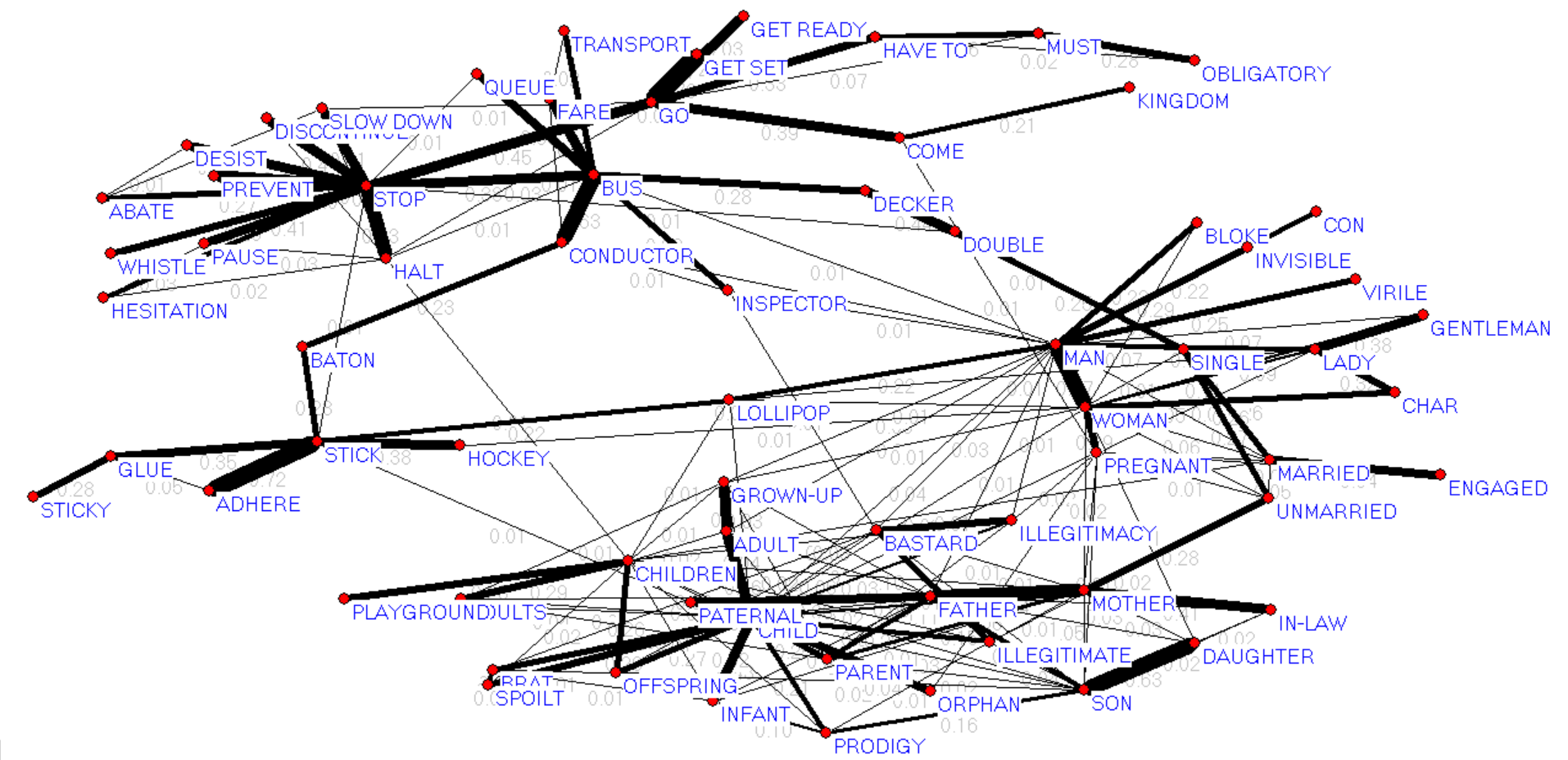}
  \caption{\footnotesize{An example of intricate structure within a cluster.}}
\end{figure}

\begin{figure}[h]
  \centering
    \includegraphics[scale=0.25]{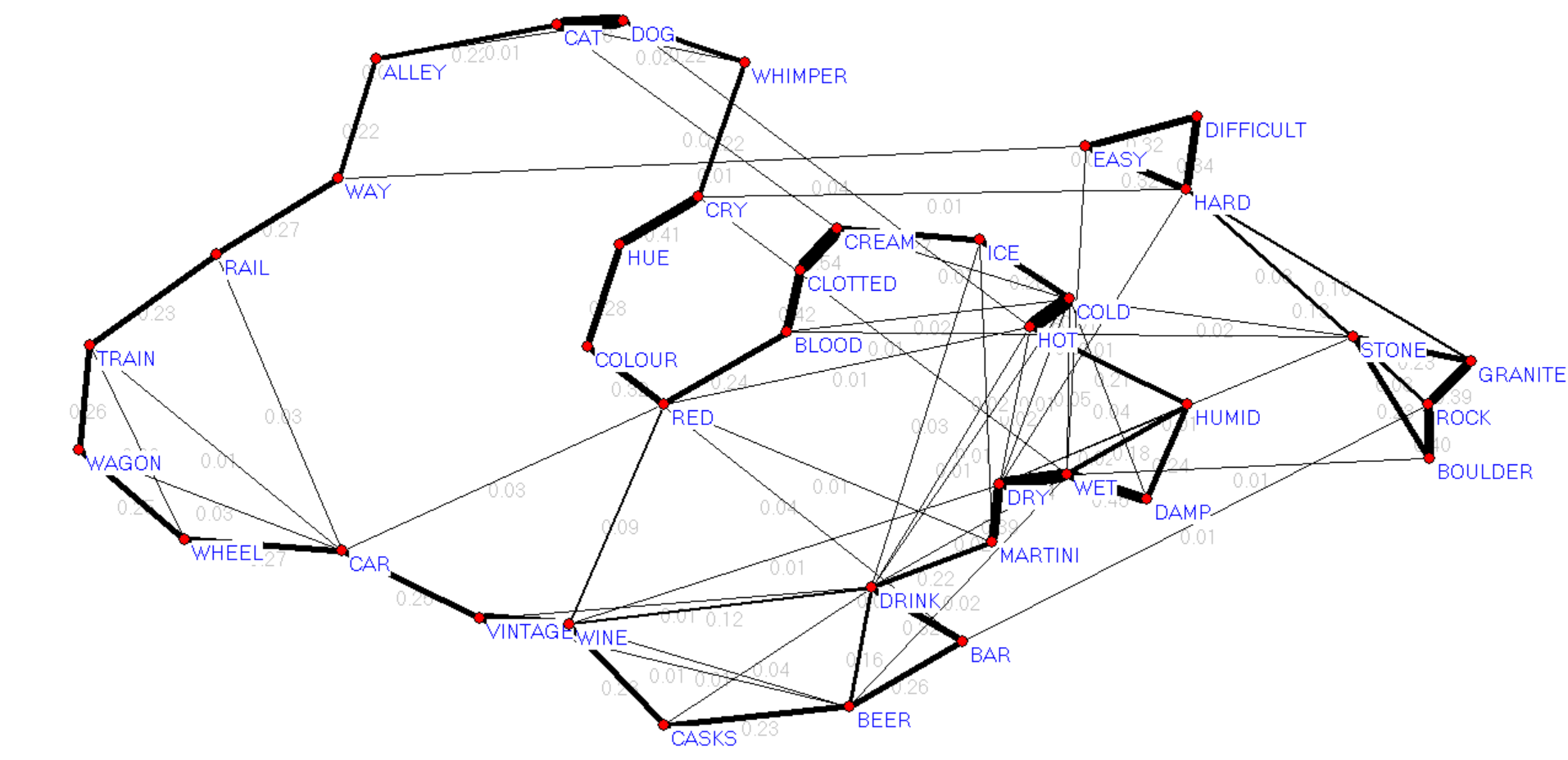}
  \caption{\footnotesize{A large loop with weak connections to other loops.}}
\end{figure}

\begin{figure}[h]
  \centering
    \includegraphics[scale=0.25]{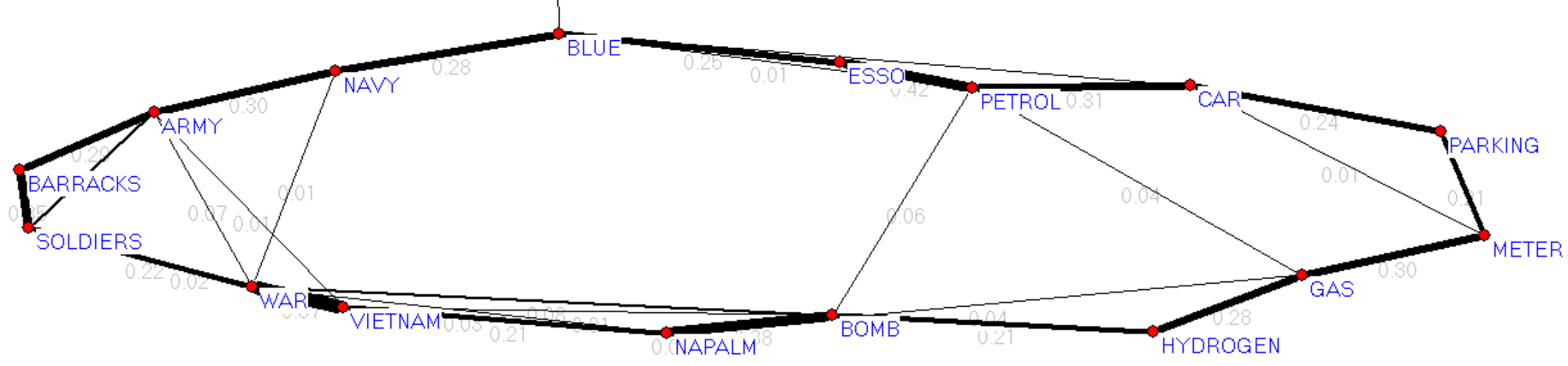}
  \caption{\footnotesize{Another large loop}}
\end{figure}

\begin{figure}[h]
  \centering
    \includegraphics[scale=0.2]{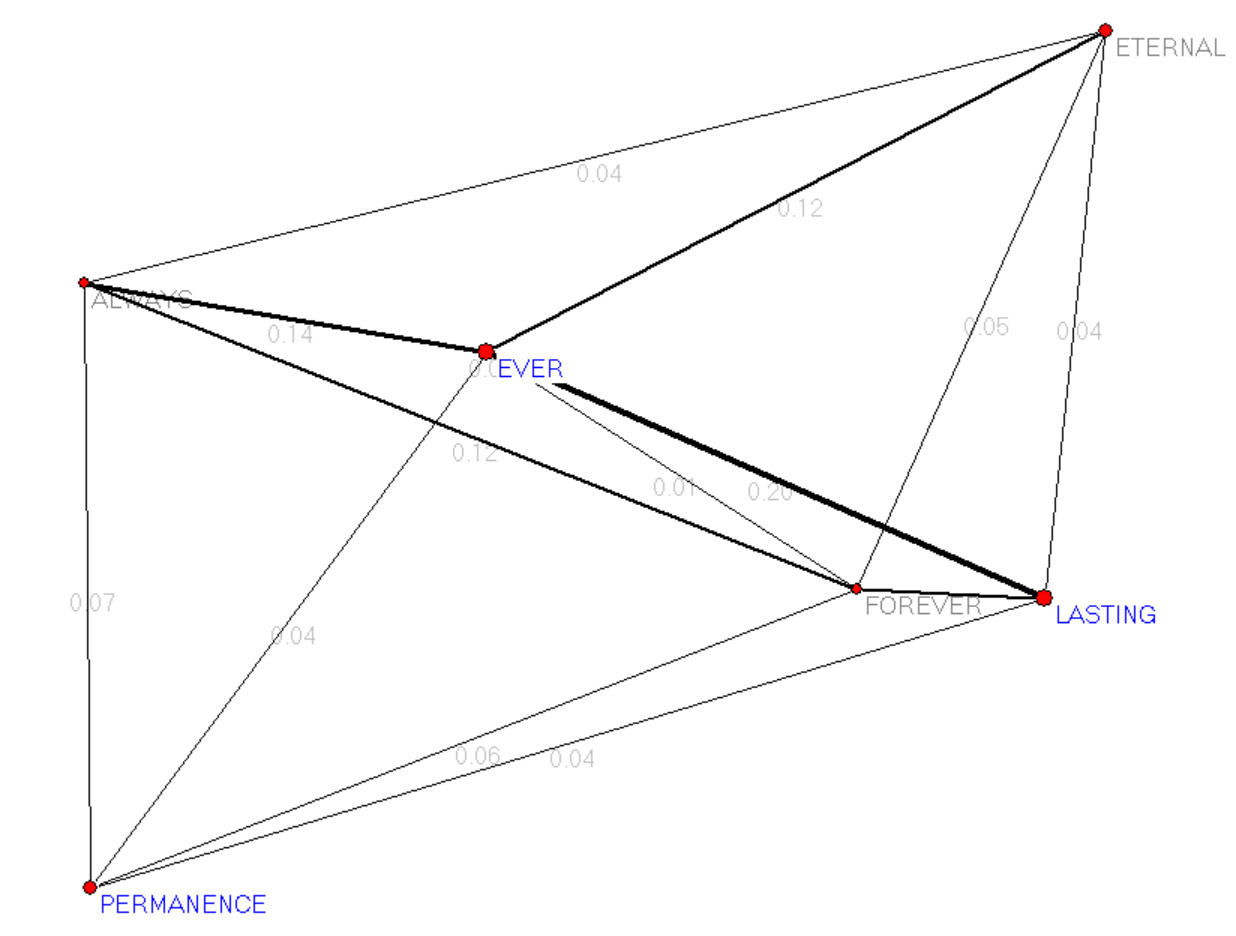}
  \caption{\footnotesize{The most persistent void found in the EAT dataset.}}
\end{figure}

\clearpage


\end{document}